\documentclass[12pt]{article}
\usepackage{graphicx}
\usepackage{subcaption}
\usepackage{hyperref}
\hypersetup{
    colorlinks=true,
    linkcolor=blue,
    filecolor=magenta,      
    urlcolor=cyan,
}
\usepackage{float}
\usepackage{listings}
\usepackage{color}
\usepackage{tikz}
\usepackage{pgfplots}
\usepackage{pgfplotstable}
\usepackage{filecontents}
\usepackage{caption}
\usepackage{subcaption}
\usepackage{geometry}
\usepackage{booktabs}
\usepackage{multirow}
\usepackage{longtable}
\usepackage{tabularx}
\usepackage{array}
\usepackage{pdflscape}
\usepackage{rotating}
\usepackage{lscape}
\usepackage{setspace}
\usepackage{caption}
\usepackage{subcaption}
\usepackage{hyperref}
\usepackage{amsmath}
\usepackage{tcolorbox}
\usepackage{algorithm}
\usepackage{algorithmic}
\usepackage{bbm} 
\usepackage[backend=biber,style=authoryear,safeinputenc]{biblatex}
\usepackage{amsmath,amsthm,amssymb}

\newtheorem{definition}{Definition}
\newtheorem{lemma}{Lemma}
\newtheorem{proposition}{Proposition}
\newtheorem{theorem}{Theorem}
\newtheorem{remark}{Remark}
\newtheorem{assumption}{Assumption}

\begin{document}

\title{Dynamic Decision-Making under Model Misspecification}
\author{Xinyu Dai}
\maketitle

\begin{abstract}
    In this study, I investigate the dynamic decision problem with a finite parameter space when the functional form of conditional expected rewards is misspecified. Traditional algorithms, such as Thompson Sampling, guarantee neither an $O(e^{-T})$ rate of posterior parameter concentration nor an $O(T^{-1})$ rate of average regret. However, under mild conditions, we can still achieve an exponential convergence rate of the parameter to a pseudo-truth set—an extension of the pseudo-truth parameter concept introduced by \textcite{white1982maximum}. I further characterize the necessary conditions for the convergence of the expected posterior within this pseudo-truth set. 
    Simulations demonstrate that while the maximum a posteriori (MAP) estimate of the parameters fails to converge under misspecification, the algorithm's average regret remains relatively robust compared to the correctly specified case. These findings suggest opportunities to design simple yet robust algorithms that achieve desirable outcomes even in the presence of model misspecifications.
\end{abstract}


\newpage

\section{Introduction}

Dynamic decision-making problems arise in many economic contexts, such as dynamic pricing, online advertising, investment, and treatment allocation. In these contexts, the decision maker needs to learn the true model from the data and make decisions based on the learned model. This generates a well-known exploration-exploitation trade-off (\cite{lattimore2020bandit}). Traditional methods such as Thompson Sampling (TS) are used for solving these dynamic problems are proved to be asymptotically optimal. However, these results contains implicit assumption that the underlying model is correctly specified and there is limited understanding how those methods behave for a possibly misspecified class of models.

On the other hand, econometricians have investigated model misspecification for a long time. It is well-known that the statistical inference under misspecification will lead to a pseudo-true parameter, which minimizes the Kullback-Leibler divergence between the true model and the pseudo-true model (\cite{white1982maximum}). Recent research in bayesian decision theory characterize conditions for priors which lead to a unique and informative pseudo-truth convergence (\cite{andrews2023structural}). However, it is still unclear how misspecification affects the behavior of parameter estimation in a dynamic environment.\footnote{There are also other recent research on the diagnosis and inference under, either local or global, model misspecification,  such as \textcite{mullerlocally}, \textcite{bonhomme2022minimizing} \textcite{armstrong2024adapting}, and \textcite{masten2021salvaging}. However, all of those research mainly focus on the static case.}

In this paper, we address these challenges by providing a preliminary analysis of Thompson Sampling under model misspecification. Our main contribution was to show that, even when the true data-generating mechanism \(\theta^{*}\) lies outside the assumed parameter space \(\Theta\), Thompson Sampling algorithm still produces posterior distributions that concentrate around a \emph{pseudo-truth set} \(\Theta^\dagger\) at an exponential rate. This result holds irrespective of which actions are chosen, provided the algorithm uses Bayesian updates to compare parameters' likelihood scores. Moreover, the total regret grows linearly in time \(T\), implying a constant per-period regret in the limit. Then we further derived necessary conditions under which pathwise posterior concentration on a subset of \(\Theta^\dagger\) can occur. Specifically, a subset \(S\subset \Theta^\dagger\) must be jointly closed under overshadowing and strongly connected to sustain nontrivial posterior mass in the long run.

We also performed simulations for further investigate the behavior of the parameter posterior within the psuedo-truth set. Our simulations are based on two distinct true data-generating processes (DGPs): a quadratic reward function with thresholds and a piecewise linear reward function with breakpoints. The simulation results shows that the MAP of the parameter is no longer convergnce to a single point and multiple limits can coexist. The average regret of the Thompson Sampling algorithm remains bounded by a constant, and it also shows robustness to the correctly specified case. These findings suggest opportunities to design simple yet robust algorithms that can be adaptive to model misspecifications. 

\paragraph{Literature Review} The literature on dynamic decision-making under model misspecification is still limited. Several current research on this topic is done by \textcite{adusumilli2021risk} and \textcite{fan2021diffusion}, who use the stochastic differential equation to approximate asymptotic behavior of Thompson Sampling algorithm under the bandit problem. Certain types of misspecification are considered in \autocite{fan2021diffusion,li2023dynamicselectionalgorithmicdecisionmaking}'s work such as the misspecification of the variance of the noise distribution or the endogeneity between the choices of arm and its contextual variables. However, this paper's contribution distinguishes from these research in three folds. First, we investigate a more flexible setting of misspecification of reward functions and this will plays a key role in our analysis of the algoirthms dynamics. Second, our setting is adapt to continuous bandit problem with a parametric setting rather than a fully nonparametric discrete bandit problem, which could be more attractive to applied economists, who may hold certain prior knowledge of the reward function implied from the economic theory. Third, as far as we know, our investigation on misspecified dynamics decision making from the pseudo truth parameter perspective is novel to the literature. \footnote{Moreover, misspecified bandit problem is also considered in the reinforcement learning literature. Some representative works are \textcite{pmlr-v119-lattimore20a,NEURIPS2021_177db6ac,NEURIPS2020_84c230a5}. However, most of those works restrict their attention to a fully nonparametric discrete bandit problem and focus on the regret rate analysis instead of the behavior of a misspecified model parameter.} 

Our paper can also be related to the economic theory literature of misspecified learning. Economists have long been interested in the impact of model misspecification on the behavior of agents (\cite{https://doi.org/10.3982/ECTA12609,fudenberg2023misspecifications,ba2023robustmisspecifiedmodelsparadigm}). Though we share some similarity in terms of the misspecification setting, those research mainly assume the agent will take an \textit{ad hoc} decision rule in a specific generic of economic context, such as dynamic game or macroeconomic model (\cite{ESPONDA2021105260,murooka2023higher}). The analysis are interested in the economic insights can be generated from that context accordingly. In contrast, our paper focus more on parameter inference from a sophisticated statistician perspective, with purpose on how the pseudo-truth parameter can be characterized and how the algorithm can be designed to be robust to the misspecification. In that sense, our work rich the current economic theory considering the player as statisticians by taking the misspecification in to account (\cite{liang2019games,salant2020statistical}). (may need to be more specific)


In the following sections, we will investigate the impact of model misspecification on the decision-making process through simulation results. We will compare the performance of Thompson Sampling under different levels of model misspecification, as well as different levels of noise in the rewards.

\section{Model Setup}

We consider a dynamic decision problem taking place over discrete time steps \(t = 0, 1, 2, \ldots\). At each time \(t\):

\begin{enumerate}
    \item The decision maker observes a covariate \(X_t \in \mathcal{X}\).
    \item Based on the current history \(H_t = \bigl(A_0, X_0, R_0, \ldots, A_{t-1}, X_{t-1}, R_{t-1}\bigr)\), the decision maker chooses an action \(A_t \in \mathcal{A}\).
    \item A reward \(R_t\) is then realized. Its conditional distribution is assumed (by the model) to be given by a density
    \begin{equation}\label{eq:reward_model}
        R_t \;\sim\; f_{\theta}\bigl(\,\cdot \mid A_t, X_t\bigr)
    \end{equation}
    where \(\theta\in\Theta\subseteq\mathbb{R}^d\) is a parameter.
\end{enumerate}

\noindent
Here, \(\Theta\) is a finite set of candidate parameters. However, the \emph{true} parameter governing the rewards, denoted \(\theta^*\), may lie outside \(\Theta\). Formally:

\begin{assumption}[Misspecification]\label{assump:misspec}
The true parameter \(\theta^*\) that generates the data does not necessarily lie in \(\Theta\). Hence, \(\theta^* \notin \Theta\).
\footnote{The dimension of \(\theta^*\) needs not match that of \(\Theta\). One may also think that the reward function is generate by another class of function $g_{\theta^*}(.|A_t, X_t)$ which is not in the class of $f_{\theta}(.|A_t, X_t)$.}
\end{assumption}

\vspace{1em}
\noindent
\textbf{Expected Reward.} For each \(\theta \in \Theta\), define the expected reward under that parameter as 
\begin{equation}\label{eq:expected_reward}
    r_{\theta}(a, x) \;=\; \mathbb{E}_{\theta}[\,R_t \,\mid\, A_t = a,\, X_t = x\,]
\end{equation}

If \(\theta\) were known and correct, the optimal action at \((a,x)\) would be \(\arg\max_{a \in \mathcal{A}} r_{\theta}(a,x)\).

\vspace{1em}
\noindent
\textbf{Policies and Regret.} 
A \emph{decision rule} or \emph{policy} \(\mu\) is a mapping from histories \(\mathcal{H}_t\) to an action in \(\mathcal{A}\). Suppose an \emph{oracle} policy \(\mu_{\theta}\) that always uses the parameter \(\theta\) to select the best action. Its expected reward from time \(0\) to \(T-1\) is
$$
\mathbb{E}_\theta^{\mu_{\theta}}\Bigl[\sum_{t=0}^{T-1} r_{\theta}(A_t, X_t)\Bigr]
$$
A policy \(\mu\) that must learn from data could yield a (possibly lower) total expected reward,
$$
\mathbb{E}_\theta^{\mu}\Bigl[\sum_{t=0}^{T-1} r_{\theta}(A_t, X_t)\Bigr]
$$
We define the \emph{average expected regret} of \(\mu\) over \(T\) periods as
\begin{equation}\label{eq:average_regret} 
    AR(\mu) \;=\; \frac{1}{T} \left|\,
    \mathbb{E}_\theta^{\mu_\theta}\!\Bigl[\sum_{t=0}^{T-1} r_{\theta}(A_t, X_t)\Bigr]
    \;-\;
    \mathbb{E}_\theta^{\mu}\!\Bigl[\sum_{t=0}^{T-1} r_{\theta}(A_t, X_t)\Bigr]
    \right|
\end{equation}

\vspace{1em}
\noindent
\textbf{Posterior Updating.}
Given a prior distribution \(\pi_0(\theta)\) over \(\Theta\), the posterior after observing history \(H_t\) is
\begin{equation}\label{eq:posterior}
    \pi_t(\theta) \;=\;
    \frac{\mathcal{L}_\theta\bigl(H_t\bigr)\,\pi_0(\theta)}
    {\sum_{\gamma \in \Theta} \mathcal{L}_\gamma\bigl(H_t\bigr)\,\pi_0(\gamma)},
\end{equation}
where \(\mathcal{L}_\theta(H_t)\) is the likelihood of observing \(H_t\) under \(f_{\theta}\). 

\vspace{1em}
\noindent
\textbf{Thompson Sampling (TS).} 
A well-known Bayesian learning algorithm for exploration–exploitation is Thompson Sampling, which proceeds as follows:

\begin{enumerate}
    \item Update the posterior \(\pi_t(\theta)\) from \eqref{eq:posterior}.
    \item Sample a parameter \(\theta_t \sim \pi_t(\theta)\).
    \item Select the action 
    $$
    A_t \;=\; \arg\max_{a \in \mathcal{A}} \; r_{\theta_t}(a, X_t) 
    \;=\; \arg\max_{a \in \mathcal{A}} \; \mathbb{E}_{\theta_t}[R_t \mid a, X_t]
    $$
    \item Observe the reward \(R_t\), then update the posterior to \(\pi_{t+1}(\theta)\) accordingly.
\end{enumerate}

\vspace{1em}
\noindent
\textbf{Kullback-Leibler (KL) Divergence.}
For two probability measures \(v_{\theta}\) and \(v_{\gamma}\), we define their KL divergence as

\begin{equation}\label{eq:KL_divergence}
    \mathcal{K}(v_{\theta} \,\mid\, v_{\gamma})
    \;=\;
    \mathbb{E}_{\theta} \Bigl[\,
    \log\!\bigl(\tfrac{dv_{\theta}(a)}{dv_{\gamma}(a)}\bigr)
    \Bigr]
    \;=\;
    \mathbb{E}_{\theta}\!\Bigl[\,
    \log \frac{f_{\theta}(\cdot \mid a, x)}{f_{\gamma}(\cdot \mid a, x)}
    \Bigr].
\end{equation}

This will be useful for analyzing how quickly the posterior discards parameters that fit the data poorly.

\section{Numerical Examples}

We firstly consider a simple example where the decision maker will assume their reward function will come from a class of quadratic function with a Gaussian noise. It can be parameterized as follows:
\begin{equation}
    R_t = f_{\theta}(A_t) =  \theta_1 + \theta_2 A_t + \theta_3 A_t^2 + \epsilon_t 
\end{equation}

where $\epsilon_t \sim N(0, \sigma^2)$ is the noise term, The parameter space is defined as $\theta = (\theta_1, \theta_2, \theta_3) \in \Theta$. \footnote{Since we currently focus on the misspecification of expected reward function, we assume the variance of the noise term is known for simplicity.} This model can be used to represent a wide range of economic problems, such as dynamic pricing, investment, and treatment allocation. For example, in a dynamic pricing problem, such a quadratic reward function can be interpreted as the revenue generated by setting a price $A_t$ at time $t$, which is solved from a linear demand and supply system. The decision maker aims to learn the true parameter $\theta^*$ and make decisions that maximize the expected reward.

In such a setting, the oracle decision rule should be $\mu(t) = \mu* = -\frac{\theta_2}{2\theta_1}$ if we know the true parameter. However, when the parameter is unknown, a classical solution is using a heuristic algoirthms to balance the exploration and exploitation trade off. For example, Thompson Sampling algorithm is a Bayesian algorithm that samples the parameter from the posterior distribution and selects the action that maximizes the expected reward. The specific implementation of the Thompson Sampling algorithm is shown in Algorithm 1.

\begin{algorithm}
    \caption{Thompson Sampling (TS) Algorithm}
    \begin{algorithmic}[1]
    \STATE \textbf{Initialize:} Prior distribution $\pi(\theta)$
    \FOR{each time step $t = 1,2,\dots$}
        \STATE Sample $\theta_t \sim \pi(\theta)$
        \STATE Select action $A_t = \arg \max_{a} f_{\theta_t}(r \mid a, X_t)$
        \STATE Observe reward $R_t$
        \STATE Update posterior distribution $\pi(\theta \mid \{(A_s, R_s)\}_{s=1}^{t})$
    \ENDFOR
    \end{algorithmic}
\end{algorithm}

The posterior distribution of the parameter $\theta$ at time $t$ is given by \eqref{eq:posterior}. 

The existing research has shown the Thompson Sampling algorithm can achieve an $O(e^{-T})$ rate of convergence of the posterior parameter to the true parameter under the correctly specified model (\cite{kim2017thompson,banjevic2019thompson}) and the expectation of parameter posterior will converge to the true parameter with exponential rate. 

Now let us consider two examples of misspeciication on the above quadratic expected reward function's underlying assumptions. The first is the symmetry of the reward function (Example 1) and the second is unimodal of the reward function (Example 2). Those assumptions are generally not true when the demand function is nonlinear. In the following simulations, we will investigate how the Thompson Sampling algorithm behaves under these two types of model misspecification.

True DGP Example 1:
\begin{equation}
    R_t = f_{\theta, \delta}(A_t) = \begin{cases}
    \beta_1 (A_t - \alpha_1) + \alpha_2 + \epsilon_t, & \text{if } A_t < \alpha_1 \\
    \beta_2 (A_t - \alpha_1) + \alpha_2 + \epsilon_t, & \text{otherwise}
    \end{cases}
\end{equation}

Let $\{\epsilon_t\}_{t \geq 1}$ be a sequence of i.i.d. random variables such that $\epsilon_t \sim \mathcal{N}(0, \sigma^2)$. The parameter space is given by  
\begin{equation}
    \Theta' = \left\{ \theta = (\alpha_1, \alpha_2, \beta_1, \beta_2) \in \mathbb{R}^4 \mid \beta_1 \neq \beta_2 \right\}.
\end{equation}



True DGP Example 2:

\begin{equation}
    R_t = f_{\theta, \delta}(A_t) = \begin{cases}
    \theta_1 + \theta_2 A_t + \theta_3 A_t^2 - M + \epsilon_t, & \text{if } -\frac{\theta_3}{2\theta_2} - \delta \leq A_t \leq -\frac{\theta_3}{2\theta_2} + \delta \\
    \theta_1 + \theta_2 A_t + \theta_3 A_t^2 + \epsilon_t, & \text{otherwise}
\end{cases}
\end{equation}

where $\epsilon_t \sim N(0, \sigma^2)$ is the noise term. The parameter space is defined as:
\begin{equation}
    \Theta' = \left\{ \theta = (\theta_1, \theta_2, \theta_3, \delta) \in \mathbb{R}^4 \mid \delta > 0, \theta_1 < 0 \right\}
\end{equation}

\begin{figure}[H]
    \centering
    \includegraphics[width=8cm]{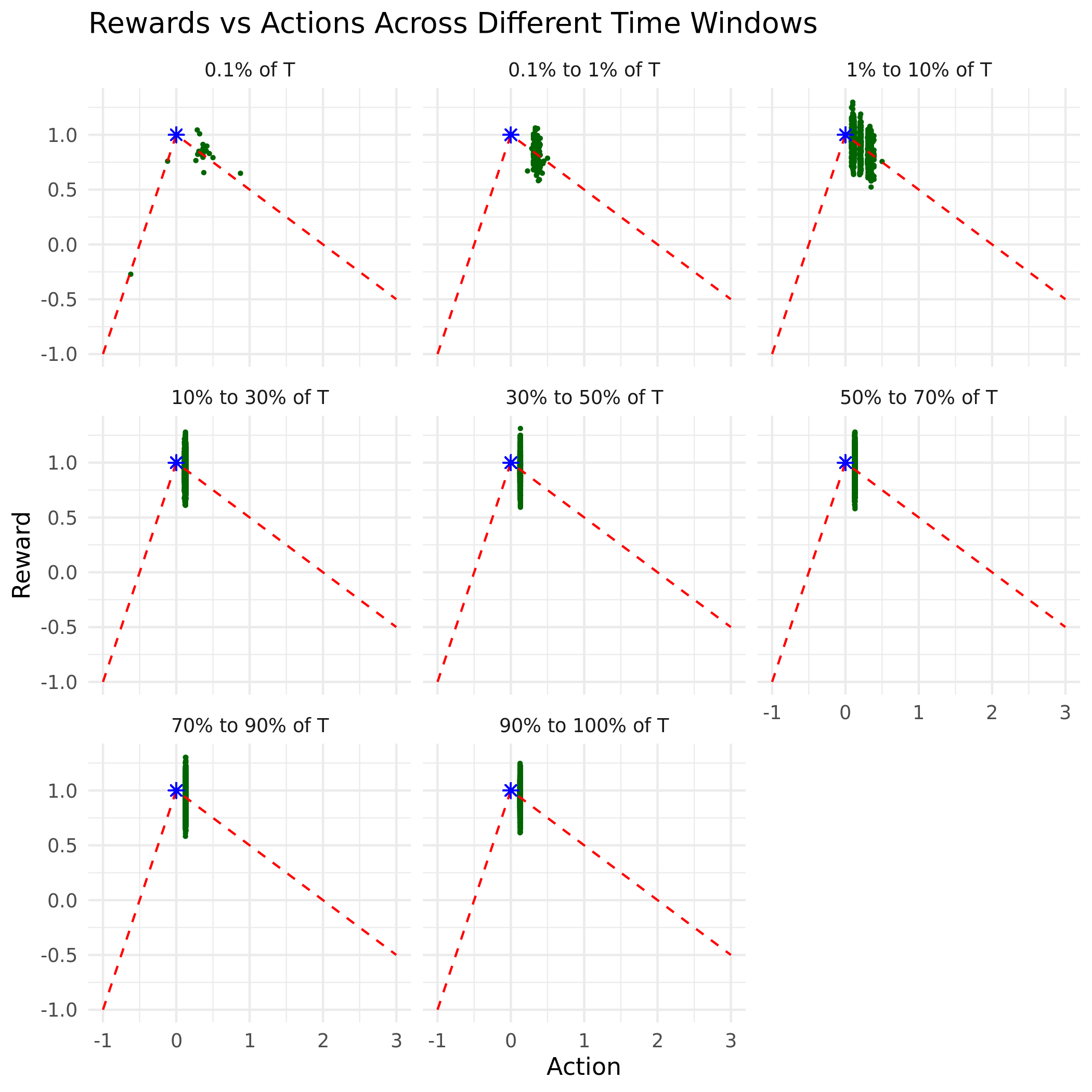}
    \includegraphics[width=8cm]{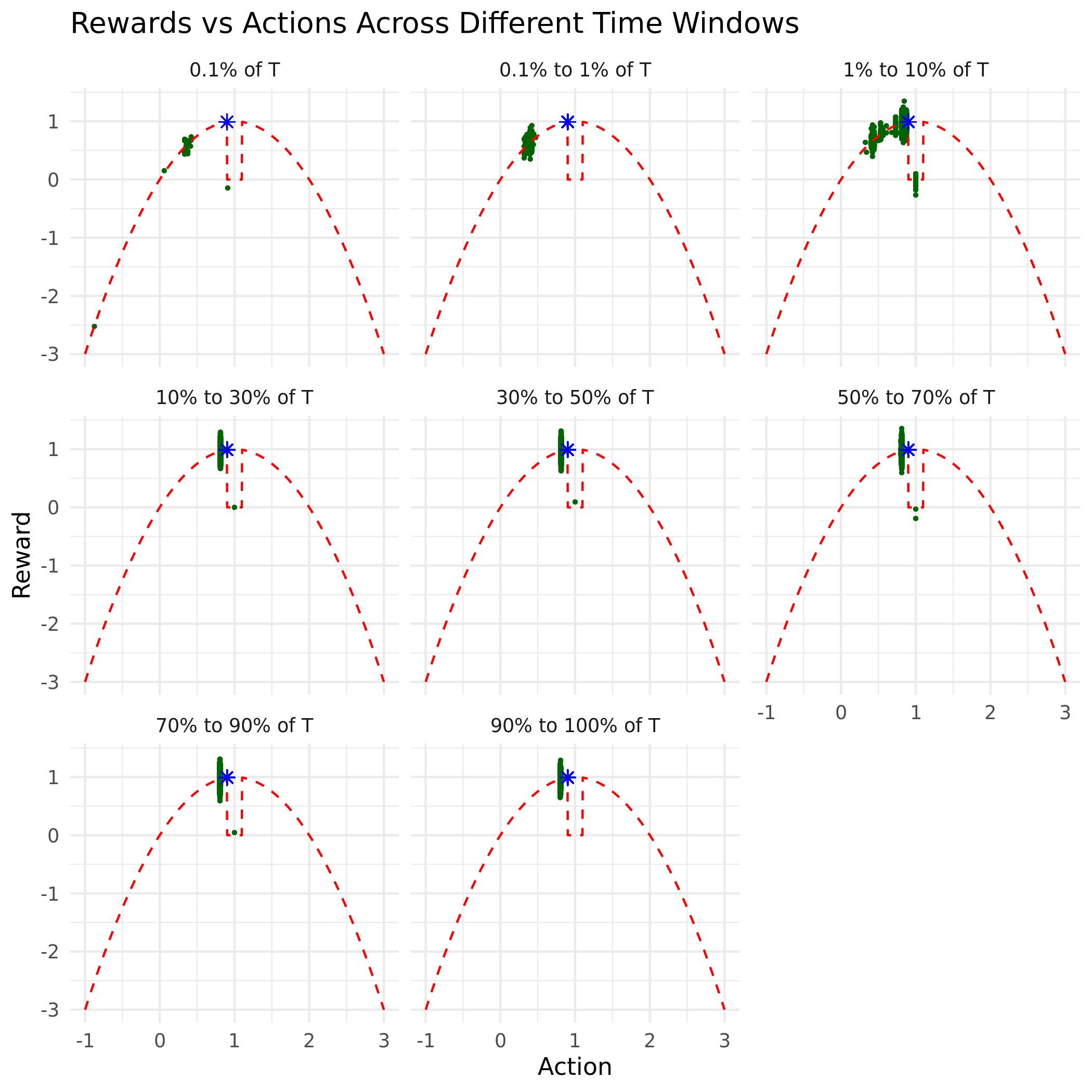}
    \caption{This is an illustration of the rewards and action relation over time under a Thompson Sampling algoirthms. The red dashed line represents the true expected reward function. The green dots represent the action taking and the rewards observed at each time period. The left panel shows the case for example 1 with the true parameter $(\alpha_1, \alpha_2, \beta_1, \beta_2) = (0,1,2,\frac{1}{2})$. The right panel shows the case for example 2 with the true parameter $(\theta_1, \theta_2, \theta_3, \delta) = (0, 2, -1, 0.1)$. The noise term $\epsilon \sim \mathcal{N}(0, 0.1)$ maximun time period $T$ is set to 20000.}
    \label{fig:Rewards_vs_Actions}
\end{figure}

Figure \ref{fig:Rewards_vs_Actions} shows under model misspecification, the Thompson Sampling algorithm will assign less probability to the true parameter and more probability to 'pseudo-true' parameters. Moreover, the simulation results show that the MAP of the parameter is no longer convergnce to a single point and multiple limits can coexist (Figure \ref{fig:MAP})\footnote{The results from misspecified model in Example 1 is similar to Example 2 but are not shown here, because the correct specified model uses different class of parameters, which are not directly comparable.}. In such cases, the pseudo-true parameter is not even well-defined. And since the bias of the parameter estimation persists, the decision rule will be suboptimal. This results in a constant average regret (or linear accumulative regret) rate rather than an exponential convergence rate for the Thompson Sampling algorithm (see in Figure \ref{fig:action_distribution} and Figure \ref{fig:average_regret}).

\begin{figure}[H]
    \centering
    \includegraphics[width=5cm]{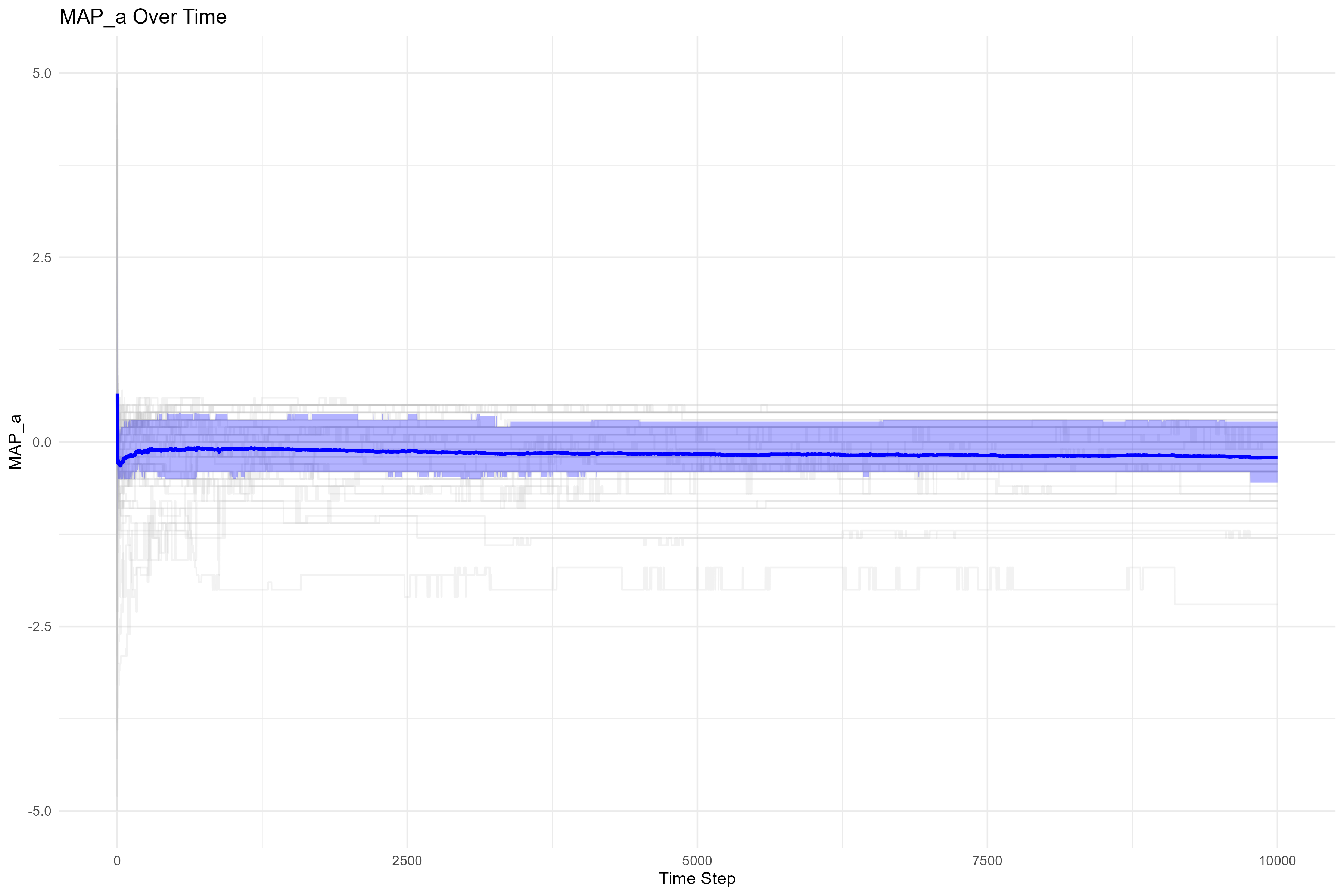}
    \includegraphics[width=5cm]{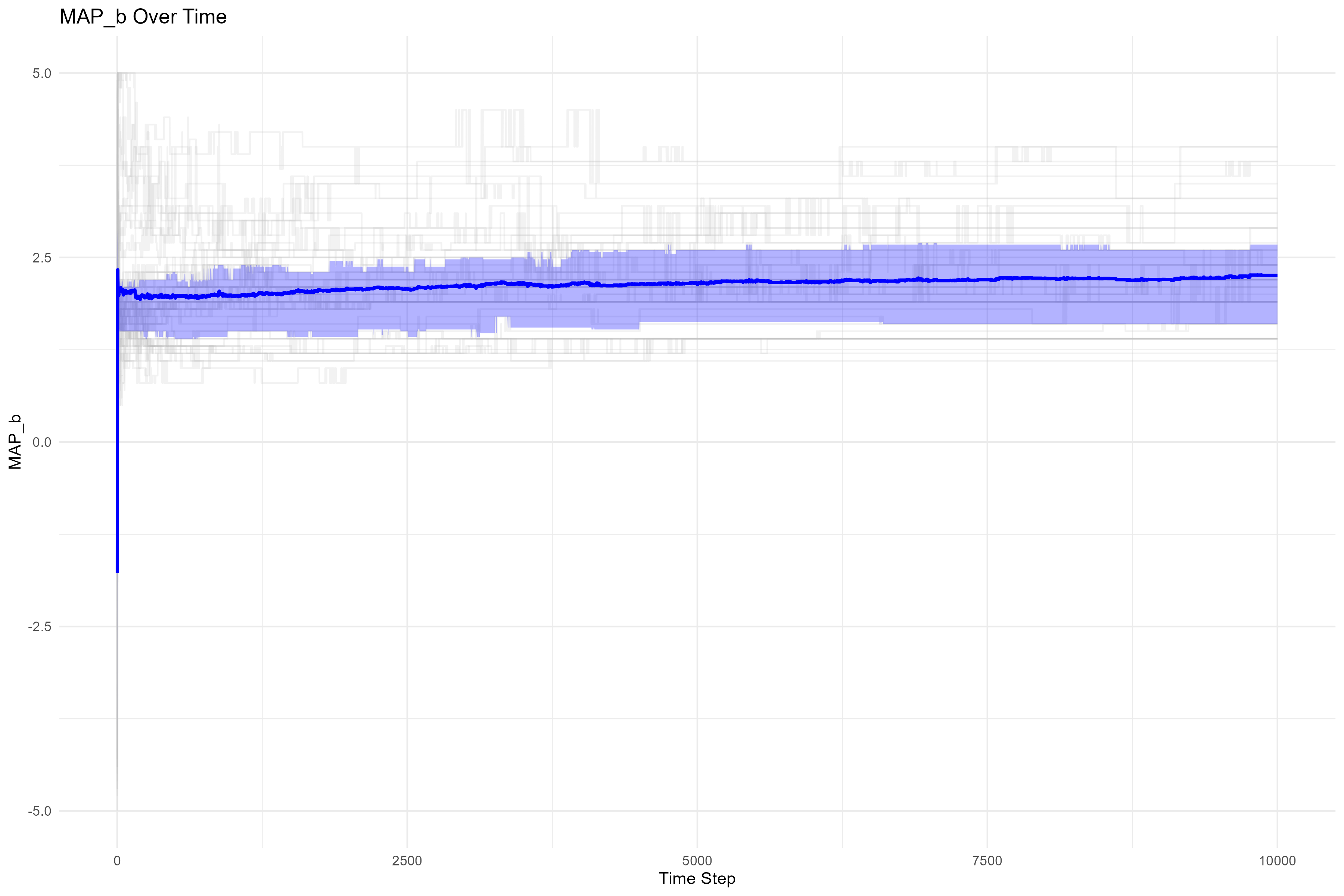}
    \includegraphics[width=5cm]{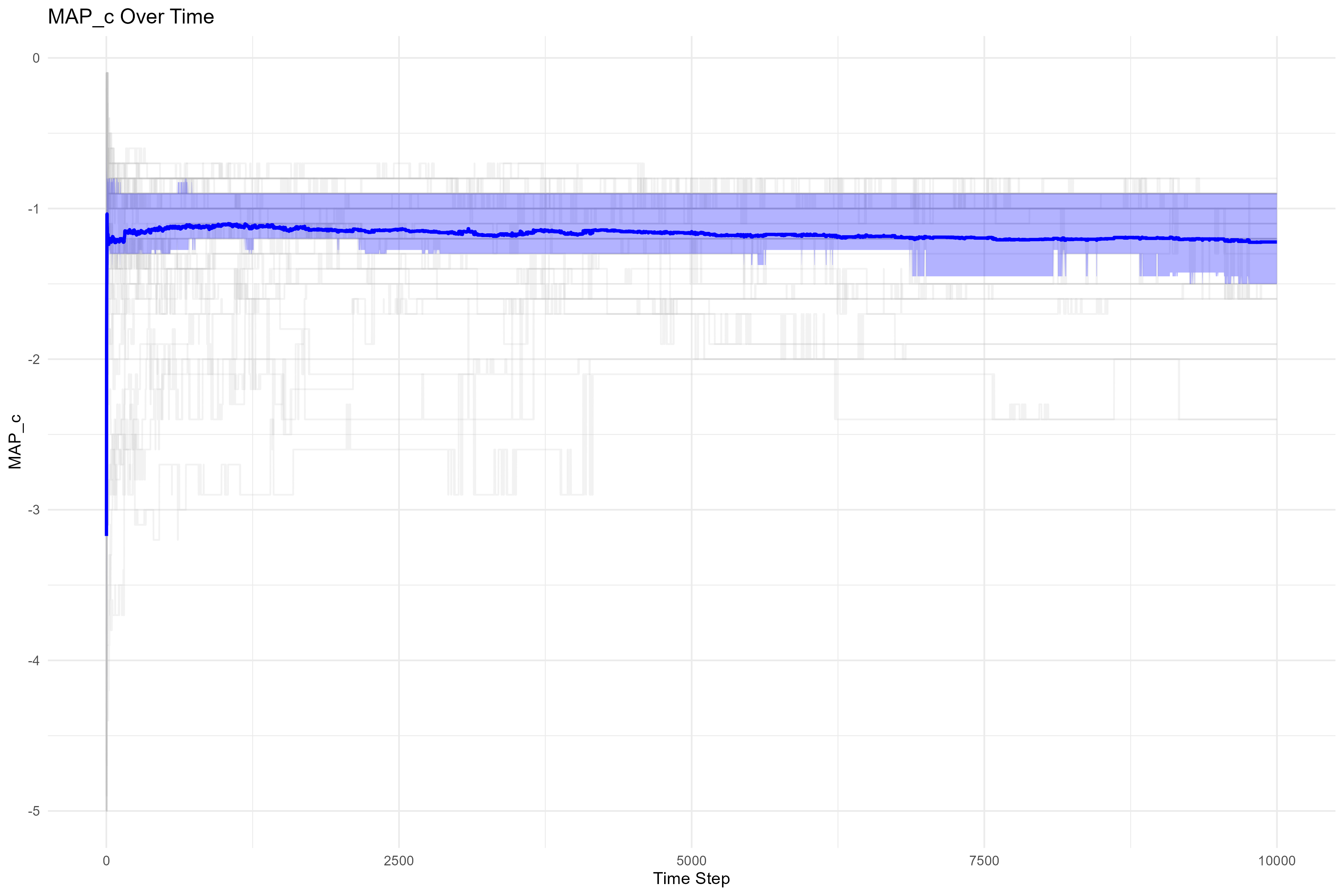}
    \\
    \includegraphics[width=5cm]{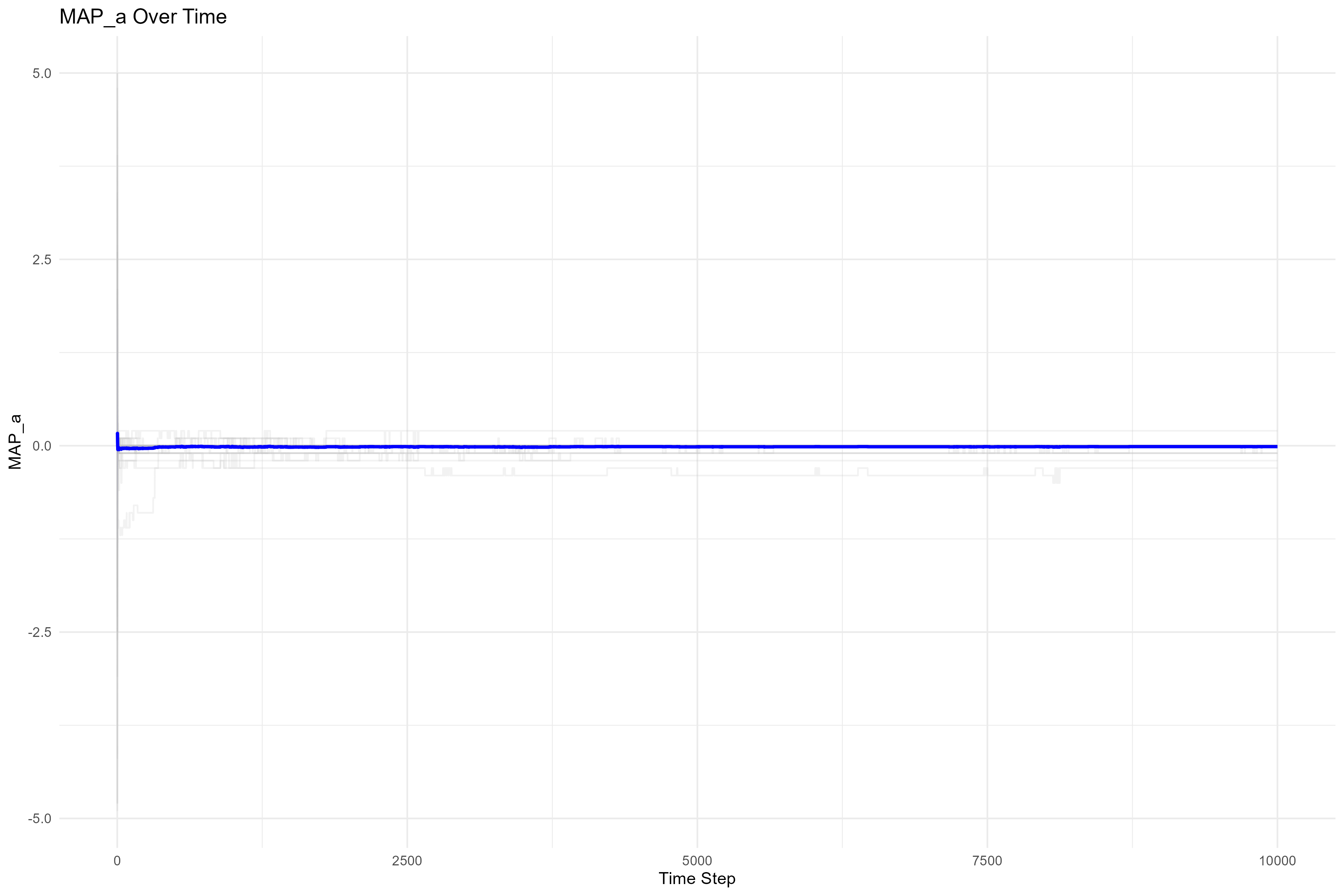}
    \includegraphics[width=5cm]{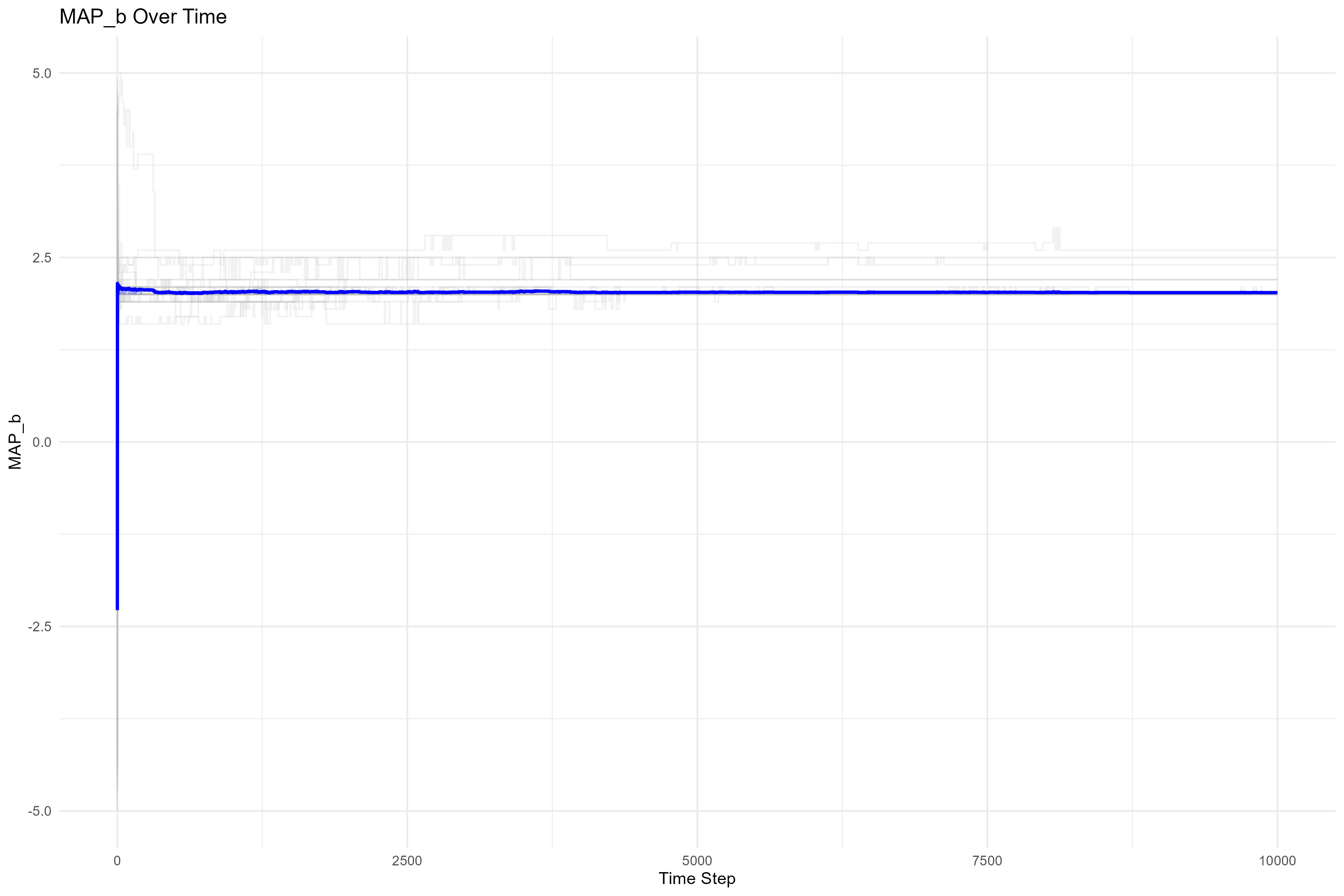}
    \includegraphics[width=5cm]{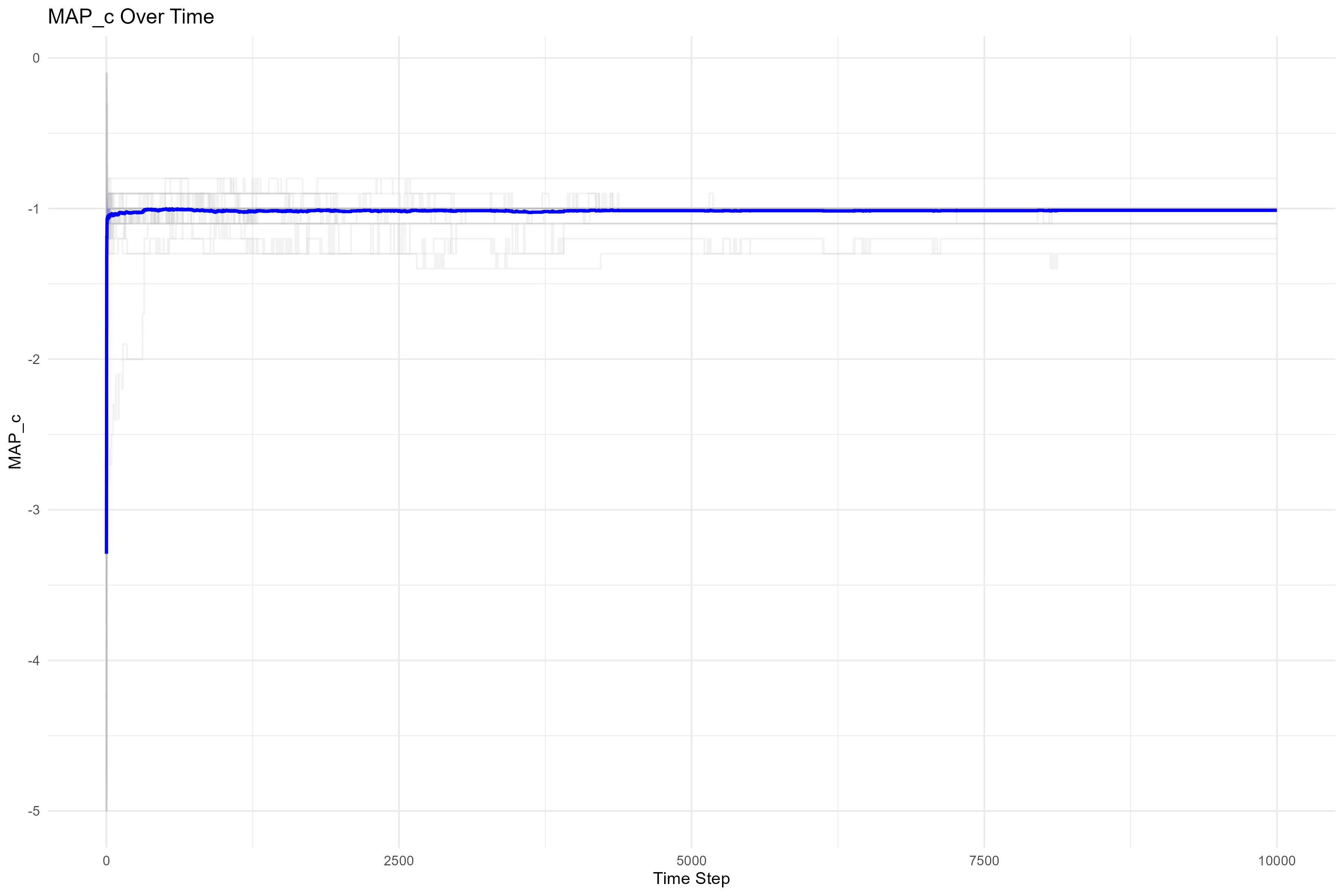}
    \caption{Maximum a Posteriori (MAP) estimates of the parameters over time. The upper row represents the misspecified model in Example 2, while the lower row corresponds to the correctly specified model. Each column shows the MAP estimates for one of the three parameters. The grey lines represent individual simulations, the blue line represents the average value over all simulations, and the shaded bands indicate the upper and lower 25\% quantiles over the simulations.}
    \label{fig:MAP}
\end{figure}

\begin{figure}[H]
    \centering
    \includegraphics[width=8cm]{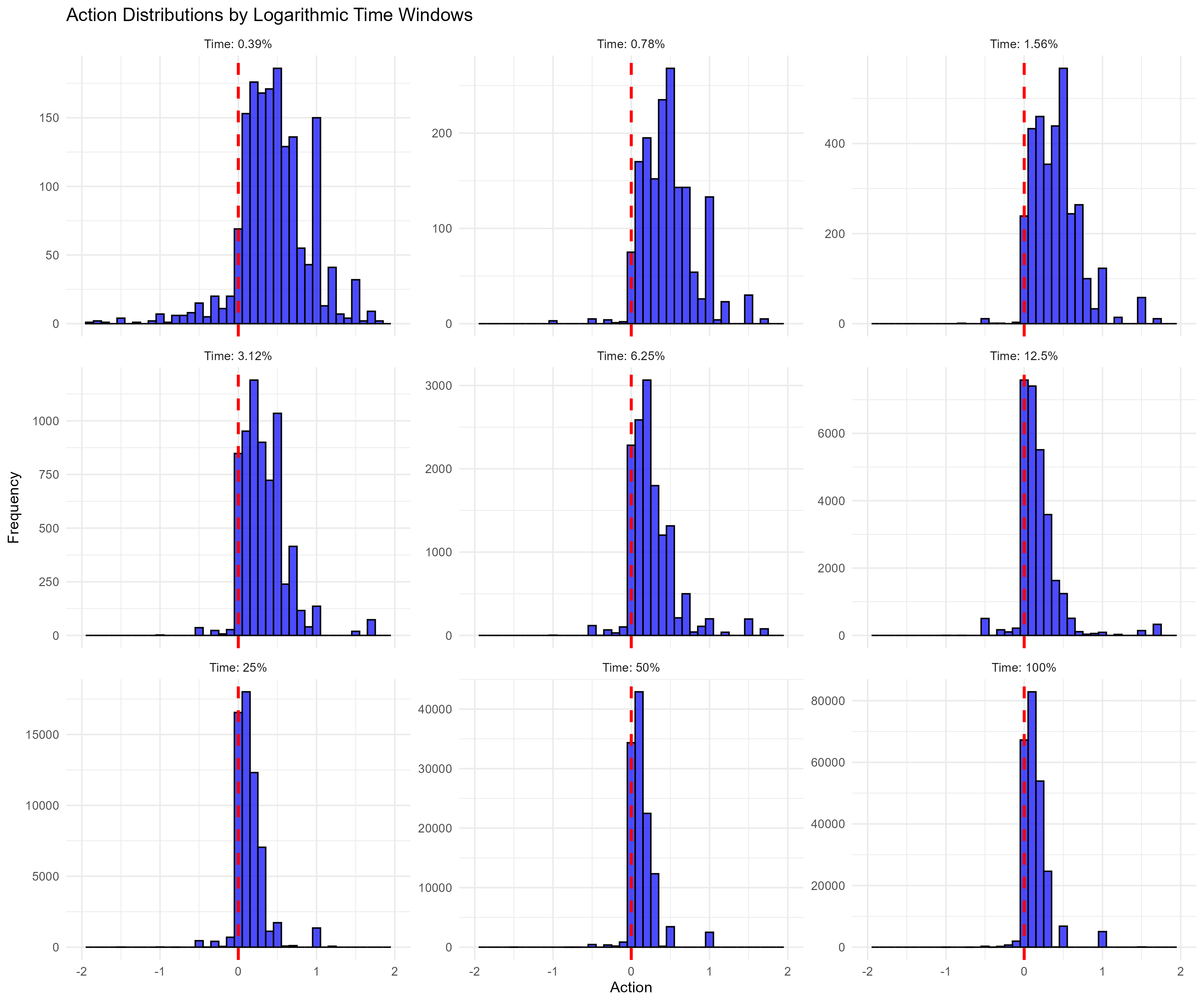}
    \includegraphics[width=8cm]{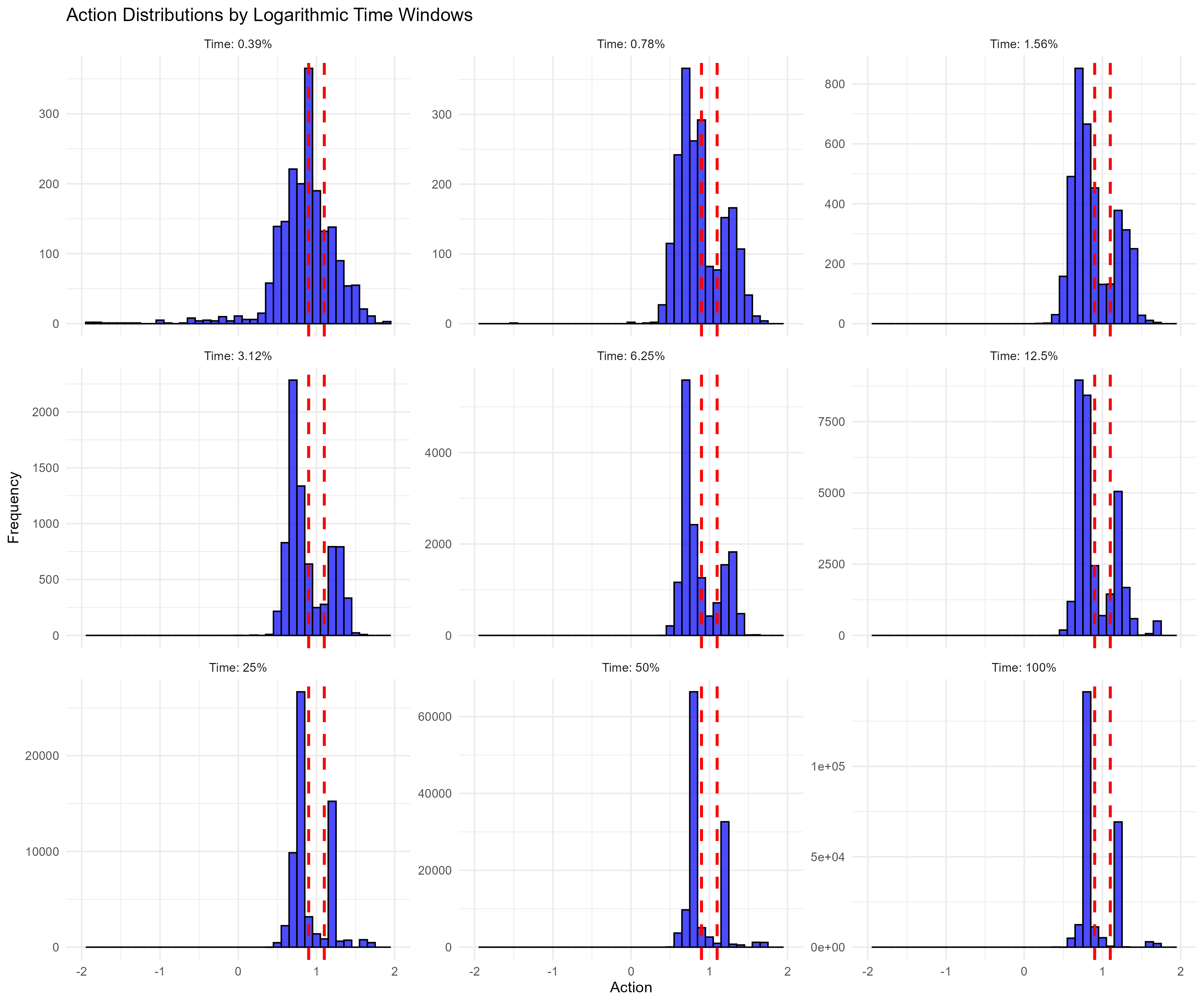}
    \caption{Action taking distribution histogram over different time windows. The red dashed line represents the optimal action, and the blue columns represent the frequency of the corresponding actions. The left panel corresponds to Example 1, and the right panel corresponds to Example 2.}
    \label{fig:action_distribution}
\end{figure}

\begin{figure}[H]
    \centering
    \includegraphics[width=8cm]{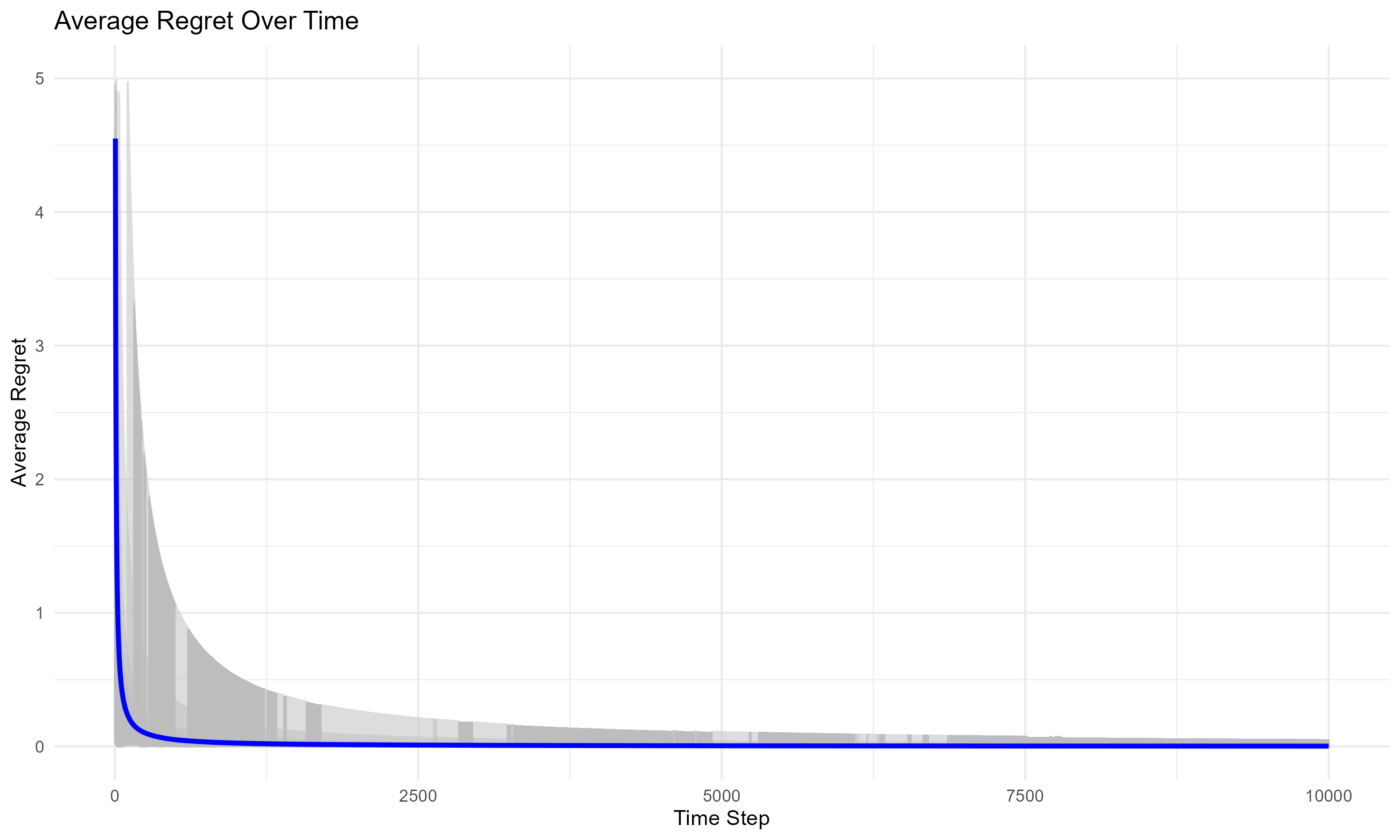}
    \includegraphics[width=8cm]{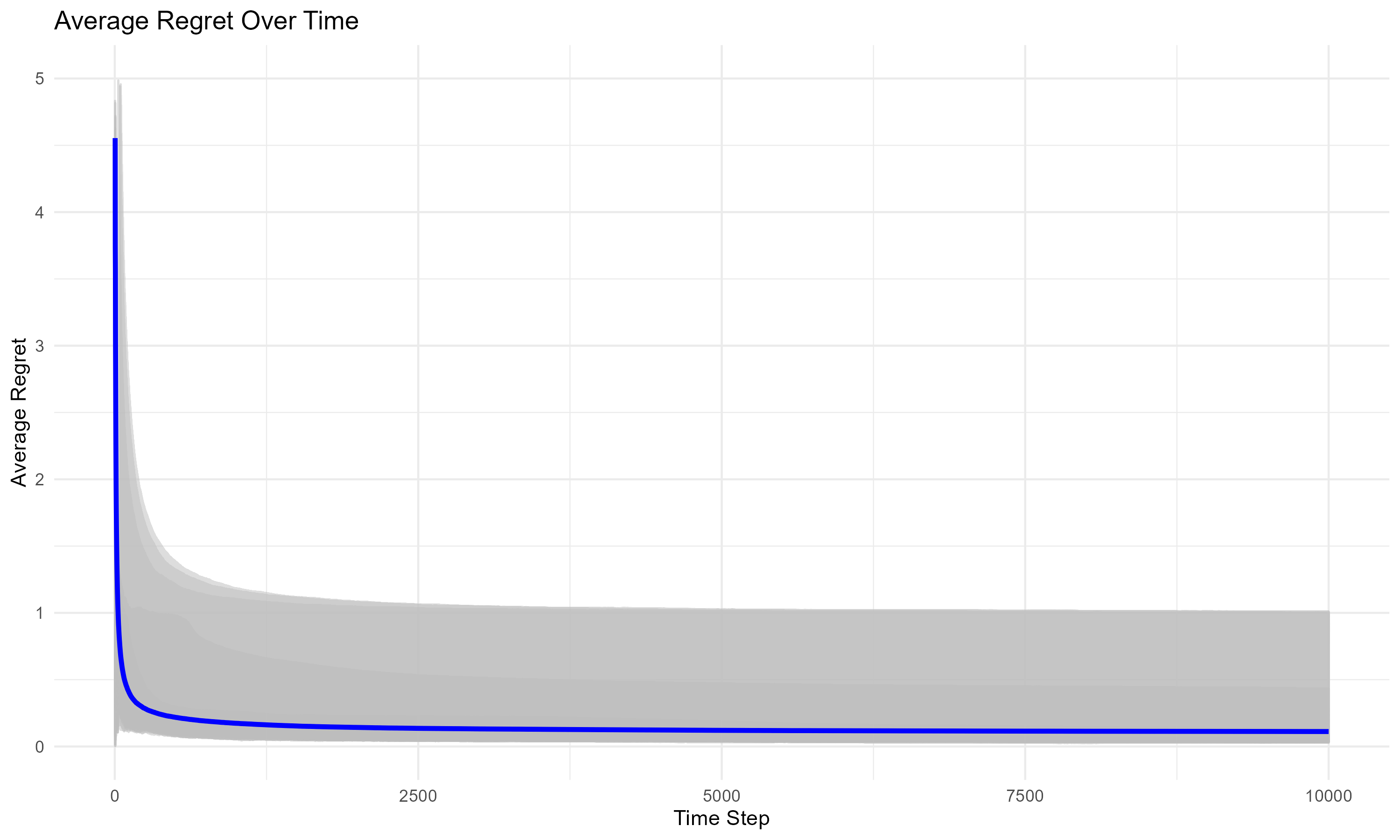}
    \caption{Average regret over time under correct (left) and misspecified (right) models. The left panel represents the correctly specified model in Example 2, while the right panel represents the misspecified model with $\delta = 0.1$. The grey band marks the upper and lower bounds of the average regret over 50 simulations, and the blue line represents the mean average regret.}\label{fig:average_regret}
\end{figure}

\section{Main Results}
Now I want to show the first main result of this paper, that is given the model misspecification, the posterior distribution of the parameter will concentrate around a pseudo-true parameter set $\Theta^\dagger$.

\textit{\textbf{Assumption 2}} (psuedo-truth set): \textit{There exist a subset $\Theta^{\dagger} \subset \mathcal{P}$, for any $\theta$ not in $\Theta^{\dagger}$, the following inequality holds:}
\begin{equation}
\sup_{\theta^{\dagger} \in \Theta^{\dagger}} \inf_{a \in \mathcal{A}} \mathcal{K}\left(v_\theta^{a} \mid v^*\right)-\mathcal{K}\left(v_{\theta^{\dagger}}^{a} \mid v^*\right) \geq \epsilon(\theta)>0
\end{equation}
for some positive $\epsilon(\theta)$. 

The intuition behind this assumption is said to be that the pseudo-truth set $\Theta^{\dagger}$ is the set of parameters such that it can not be uniformly dominated by any other parameter in the parameter space $\mathcal{P}$ in terms of the KL divergence. In other words, for any parameter which is not in the pseudo-truth set, there exists a positive gap between the KL divergence of a certain parameter and the true parameter for any action and covariate distribution.


\begin{theorem}
\label{theorem:pseudo_truth_set}
Suppose Assumption 1-2 hold, then under Thompson Sampling $\tau$, there exisits absolute constants $a_{\Theta^{\dagger}}, b_{\Theta^{\dagger}} > 0$ such that: \footnote{
One possible choice is $
2 \max \left\{ \frac{1 - \pi_0(\Theta^{\dagger})}{\pi_0(\Theta^{\dagger})}, 2(|\mathcal{P}| - |\Theta^\dagger|) \right\} $
and $b_0 = \min \left\{ \frac{\epsilon}{2}, \frac{\epsilon^2}{8d^2} \right\}$, where the $d$ is a positive constant that depends on the specific problem. See more details in the Appendix.
}

\begin{equation}
\mathbb{E}_{\theta^*}^{\tau}\left[1- \pi(\Theta^\dagger)\right] \leq a_{\Theta^{\dagger}} e^{-b_{\Theta^{\dagger}} t}
\end{equation}

\end{theorem}

\begin{proof}[Sketch of Proof] 

Recall from the \eqref{eq:posterior}, we can rewrite the posterior on $\theta$ as:

    \begin{equation}
        \begin{aligned}
            \pi_t(\theta) & =\frac{\sum_{\theta \in \theta} \mathcal{L}_\theta\left(H_t\right) \pi_0(\theta)}{\sum_{\gamma \in \mathcal{P}} \mathcal{L}_\gamma\left(H_t\right) \pi_0(\gamma)} \\
            & = \frac{ \mathcal{L}_{\theta}\left(H_t\right) \pi_0(\theta)}{\sum_{\gamma \in \mathcal{P}} \mathcal{L}_\gamma\left(H_t\right) \pi_0(\gamma)} \\
            & = \frac{1}{1+\sum_{\gamma \notin \theta} c_\gamma\left(\frac{\mathcal{L}_\gamma\left(H_t\right)}{\mathcal{L}_{\theta}\left(H_t\right)}\right)} \\
            & = \frac{1}{1+\sum_{\gamma \notin \theta} c_\gamma \exp \left(-\sum_{s=0}^t \log \Lambda_s^\gamma\right)}
        \end{aligned}
\end{equation}

Therefore, our main goal is to characterize the behavior of the term:

$$
\sum_{\gamma \notin \theta} c_\gamma \exp \left(-\sum_{s=0}^t \log \Lambda_s^\gamma\right)
$$

This quantity can be written as a sum of submartingale components, which we decompose into a martingale sequence plus a predictable, non-decreasing sequence.  

The key step is to show that the non-negative increasing sequence is bounded by a constant. The condition is true under our Assumption 2, as ensures that in expectaion the algorithm will always prefer the parameter in the pseudo-truth set $\Theta^{\dagger}$ rather than the other parameters regardless of which actions are selected.

\end{proof}

The result above shows that the posterior distribution of the parameter will concentrate exponentially around the pseudo-true parameter set $\Theta^{\dagger}$. This is a generalization of the pseudo-truth parameter of \textcite{white1982maximum} to the dynamic decision-making problem. The intuition behind this result is that the pseudo-truth set $\Theta^{\dagger}$ is the set of parameters that cannot be uniformly dominated by any other parameter in the parameter space $\mathcal{P}$ in terms of the KL divergence. And since this results are independent of the action distribution, the algoritham such as Thompson Sampling will gradually assign less probability to the parameters that are not in the pseudo-truth set $\Theta^{\dagger}$. 

Moreover, we can also get a bound on the average regret of the Thompson Sampling algorithm.
\textbf{Corollary 1} \textit{Suppose Assumption 1-2 holds, then Thompson Sampling $\tau$ average per-period regret is bounded by an absolute constant $C_{\Theta^{\dagger}}$:}
\begin{equation}
AR(T, \tau) \leq C_{\Theta^{\dagger}}
= O(1)
\end{equation}

It is not obvious that if the pseudo-truth set is tight enough. One may be more interested in characterizing the behavior of the parameter posterior within the pseudo-truth set. Intuitively, since the distribution of the action is endogenous to the algorithm itself, this recursive decision rule may depends on the specific functional form of the expeted rewards. In the following section, we will try to give several necessary conditions of the further convergence within the pseudo-truth set.

\section{Convergence within the Pseudo-Truth Set}

Proposition 1 gives a benchmark result of the convergence of the parameter posterior under model misspecification. It shows that the exponential convergence rate attributes can still be guaranteed the periphery of pseudo-truth set $\Theta^\dagger$.

However, one more interesting question will be how the parameter posterior behaves within the pseudo-truth set. To address this question, we can characterize the process as a dynamic system and apply stochastic approximation techniques. Before we move on to the formal analysis, let us first characterize the necessary conditions for the convergence of the parameter posterior within the pseudo-truth set.

\subsection{Posterior Concentration Criteria}

In this section, we investigate under what conditions the posterior distribution 
$\pi_t(\theta)$ can \emph{concentrate} on a particular parameter or subset of parameters in the pseudo-truth set $\Theta^\dagger$. 

\begin{definition}[Individual Posterior Concentration]
\label{def:IndividualConcentration}
We say the posterior concentrates on a parameter $\theta \in \Theta^\dagger$ in expectation if
$$
   \lim_{t \to \infty} \mathbb{E}_{\theta^*}\!\bigl[\pi_t(\theta)\bigr] \;=\; 1
$$
\end{definition}

It is possible, however, that no single parameter will achieve this limit. Thuse we give the following second definition. 

\begin{definition}[Collective Posterior Concentration]
\label{def:CollectiveConcentration}
We say \emph{the posterior concentrates on a subset} $S \subseteq \Theta^\dagger$ 
\emph{in expectation} if
$$
   \lim_{t \to \infty} \mathbb{E}_{\theta^*}\!\Bigl[\sum_{\theta \in S} \pi_t(\theta)\Bigr] 
   \;=\; 1
$$
\end{definition}

Recall from the proof of Proposition 1, we can rewrite the posterior on $\theta$ as:

\begin{equation}
    \begin{aligned}
        \pi_t(\theta) & =\frac{\sum_{\theta \in \theta} \mathcal{L}_\theta\left(H_t\right) \pi_0(\theta)}{\sum_{\gamma \in \mathcal{P}} \mathcal{L}_\gamma\left(H_t\right) \pi_0(\gamma)} \\
        & = \frac{ \mathcal{L}_{\theta}\left(H_t\right) \pi_0(\theta)}{\sum_{\gamma \in \mathcal{P}} \mathcal{L}_\gamma\left(H_t\right) \pi_0(\gamma)} \\
        & = \frac{1}{1+\sum_{\gamma \notin \theta} c_\gamma\left(\frac{\mathcal{L}_\gamma\left(H_t\right)}{\mathcal{L}_{\theta}\left(H_t\right)}\right)} \\
        & = \frac{1}{1+\sum_{\gamma \notin \theta} c_\gamma \exp \left(-\sum_{s=0}^t \log \Lambda_s^\gamma\right)}
    \end{aligned}
\end{equation}

If we assume uniform prior then we can ignore the constant term $c_\gamma$. 
\subsection{Partition and Ranking}

Let $\Theta$ be a finite parameter space, and let $\mathcal{A}$ be the action space. Suppose each parameter $\theta \in \Theta$ induces a reward model $f_{\theta}(a)$, and define \footnote{Please note here the $\phi$ is different from the decision rule $\mu$ in the Thompson Sampling algorithm. The former is just a mapping from the parameter space to the action space, while the latter is a mapping from the history to the action space. However, when the true parameter is known, the decision rule $\mu$ will gives actions coincides with the mapping $\phi$ for every period.}
$$
  \phi: \Theta \;\to\; \mathcal{A}^{\circ} \;\subseteq\; \mathcal{A},
  \quad
  \phi(\theta) \;=\;\arg\max_{a \in \mathcal{A}} f_{\theta}(a)
$$
Here, $\mathcal{A}^{\circ} \subseteq \mathcal{A}$ denotes the finite set of all actions that can be the $\arg\max$ of $f_{\theta}$ for some $\theta \in \Theta$.

\begin{definition}[Partition and Ranking]
\label{def:Partition}
Fix an action $a \in \mathcal{A}^{\circ}$. Let's denote a measure of ``fit'' under the true DGP as $\Delta(\theta,a)$. 
\begin{itemize}
  \item For example, $\Delta(\theta,a)$ could be defined via expected log‐likelihood ratios relative to the true reward law, or in a Gaussian noise model proportional to $\|g(a) - f_{\theta}(a)\|^2.$
  \item We say $\theta_1 \succ_a \theta_2$ if $\Delta(\theta_1,a) < \Delta(\theta_2,a).$ 
\end{itemize}
Thus for every $a \in \mathcal{A}^{\circ}$, the parameter space $\Theta$ can be partitioned into subsets $\Theta_a$ with the following ordering:
$$
  \Theta_a(1) \;\succ_a\; \Theta_a(2) 
  \;\succ_a\; \dots 
  \;\succ_a\; \Theta_a(N_a)
$$
where $\Theta_a(k)$ collects the parameters in $\Theta_a$ whose $\Delta(\theta,a)$‐value is the $k$‐th smallest among $\Theta_a$. Here $N_a = \#\Theta_a$ is a number depending on the action $a$.
\end{definition}

\begin{remark}[Likelihood Ratio Perspective]
Let $R_t$ be the observed reward at time $t$, and let $a_t$ be the chosen action at time $t$. 
For any two parameters $\gamma$ and $\theta$, define the likelihood ratio
$$
  \Lambda_t^{\gamma, \theta}
  \;=\;
  \frac{\mathcal{L}_{\gamma}(R_t \mid a_t)}{\mathcal{L}_{\theta}(R_t \mid a_t)}
$$
where $\mathcal{L}_{\gamma}$ is the likelihood of observing $R_t$ under $f_{\gamma}$. 
Then, the expected log‐ratio
$$
  \mathbb{E}\bigl[\log \Lambda_t^{\gamma,\theta}\bigr]
  \;\;\propto\;\;
  \Delta(\theta,a_t) \;-\;\Delta(\gamma,a_t)
$$
Hence, if $\Delta(\gamma,a) < \Delta(\theta,a)$, then repeated observations at action $a$ push the posterior mass away from $\theta$ and toward $\gamma$. 
\end{remark}

\subsection{Necessary Conditions for Individual Concentration}

Now we can state the necessary condition for the individual concentration of the parameter posterior.

\begin{proposition}[Necessary Condition for Individual Concentration]\footnote{This is also the necessary conditionf for the pathwise individual concentration of the parameter posterior. See as the disccussion in the next section.}
\label{prop:NecessaryCondition}
Let $\pi_t(\theta)$ be the posterior weight on $\theta$ at time $t$ under Thompson Sampling $\tau$. Suppose for some $\theta^* \in \Theta$, we have
$$
  \lim_{t\to\infty} \mathbb{E}\bigl[\pi_t(\theta^*)\bigr]
  \;=\; 1
$$
Then $\theta^*$ must lie in $\Theta_{\phi(\theta^*)}(1)$, that is, $\theta^*$ is one of the best‐fitting parameters within its own preferred action $\phi(\theta^*)$. Formally,
$$
  \Delta(\theta^*,\,\phi(\theta^*))
  \;=\;
  \min_{\gamma \in \Theta_{\phi(\theta^*)}}
  \Delta(\gamma,\,\phi(\theta^*))
$$
\end{proposition}

\begin{proof}[Sketch of Proof]
Let $a^* = \phi(\theta^*)$. 
If $\theta^*$ \emph{did not} belong to the top rank $\Theta_{a^*}(1)$, then there would exist some $\gamma \in \Theta_{a^*}(1)$ such that $\Delta(\gamma,a^*) < \Delta(\theta^*,a^*)$. 
In that case, every time action $a^*$ is selected (infinitely often if $\theta^*$ were to dominate the posterior), the likelihood ratio $\Lambda_t^{\gamma,\theta^*}$ would, on average, exceed 1, pushing posterior mass toward $\gamma$. 
Hence, $\pi_t(\theta^*)$ cannot converge to 1 in expectation. 
Thus $\theta^*$ must be in $\Theta_{a^*}(1)$. 
\end{proof}

\begin{remark}[Connection with Psedo-truth set]
    It is easy to notice that the Psedo-truth set defined in the Assumption 2 is a largest subset of the set of parameters that satisfies the necessary condition in \ref{theorem:pseudo_truth_set}. That is any parameter never lies in the first class partitions will not belongs to the Psedo-truth set and we can have $\Theta^\dagger = \cup_{a \in \mathcal{A}^\circ} \Theta_a(1)$
\end{remark}

If the according $\Theta(1)$ is a singleton set, then the equilibrium condition is also sufficient. However, unfortunately, this is usually not the case when the the parameter space has higher dimension than the action space. For instance, in our example 2, we can find infinite parameter vector such that the function $f_\theta$ goes through a fixed point $(a, r)$ on the $\mathbb{R}^2$ plane. In general, the collective posterior concentration is a more common outcome than the individual concentration.

\subsection{Overshadowing and Graph Representation}

The collective posterior concentration can happen in quite different ways. For example, it could be the case that multiple single parameters can be concentrated by the posterior distribution separately (Figure \ref{fig:sub1}). In that case all those parameters satisfies the necessary condition for individual concentration. It could also be the case that the posterior distribution concentrates on a subset of parameters (Figure \ref{fig:sub2}). Then the parameter sample from the posterior distribution will cycle among the subset of parameters.

\begin{figure}[H]
    \centering
    \begin{subfigure}[b]{0.4\textwidth}
        \includegraphics[width=\textwidth]{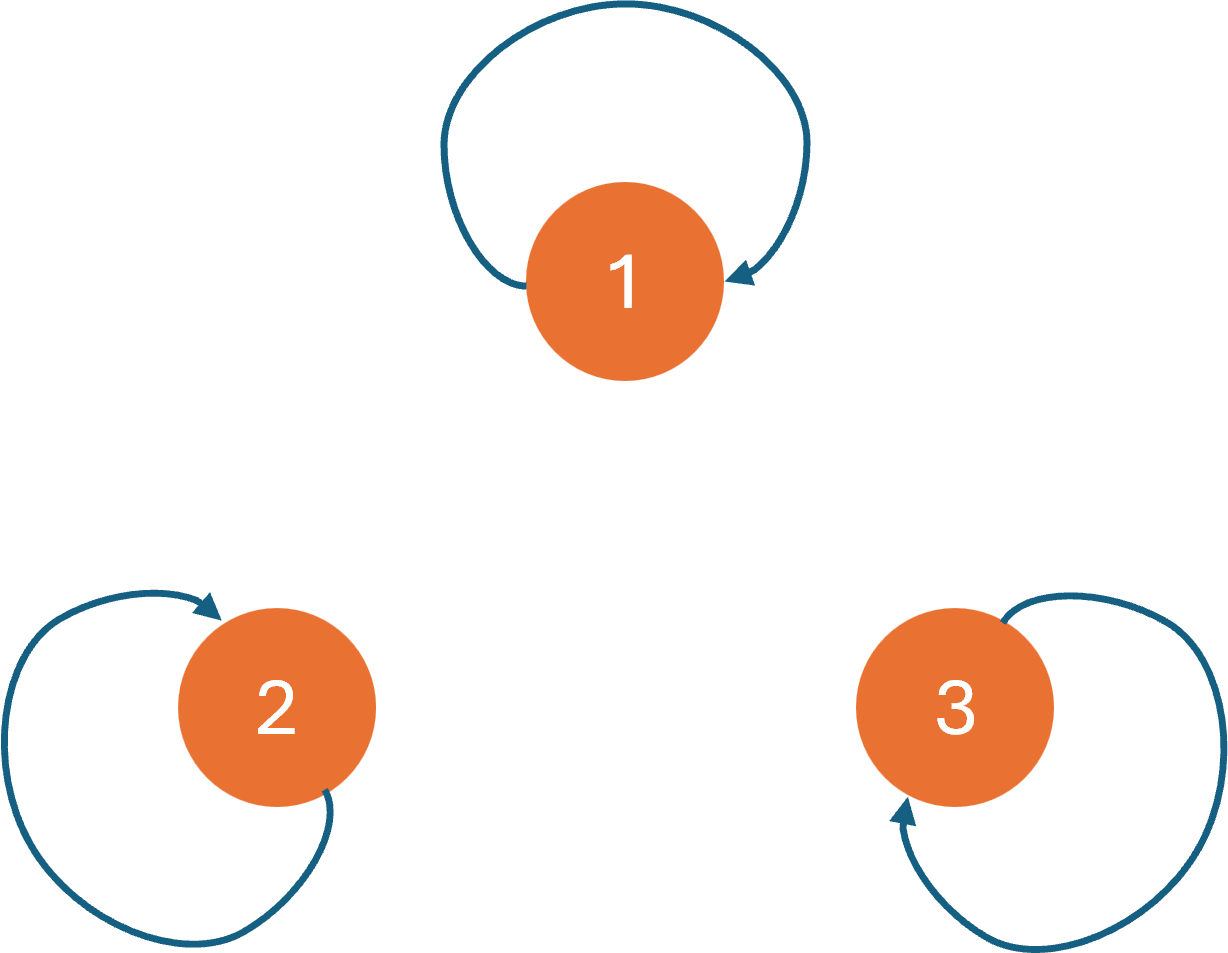}
        \caption{}
        \label{fig:sub1}
    \end{subfigure}
    \quad
    \begin{subfigure}[b]{0.4\textwidth}
        \centering
        \includegraphics[width=0.5\textwidth]{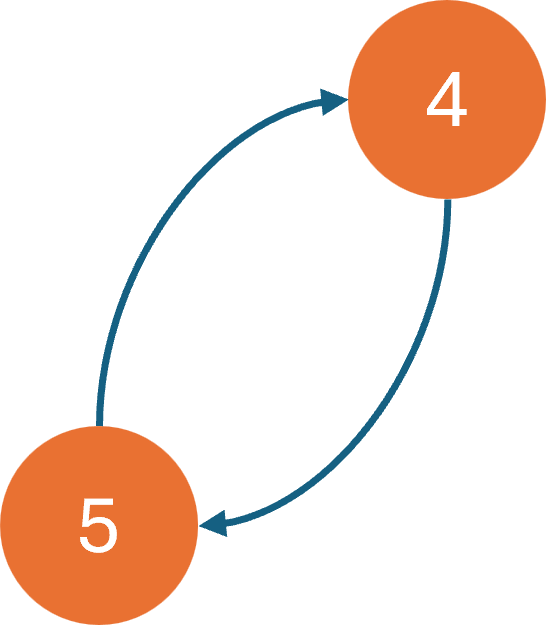}
        \caption{}
        \label{fig:sub2}
    \end{subfigure}
    \caption{Two different cases of the collective posterior concentration.}
    \label{fig:collective_concentration}
\end{figure}

If the different patterns like in Figure \ref{fig:sub1} and Figure \ref{fig:sub2} can coexist, then at the expectation level, then we can only guarantee the expected posterior distribution will concentrate on a the union of all those possible parameters (parameter 1-5 in Figure \ref{fig:collective_concentration} should be all included).

\begin{definition}[Overshadowing Edge]
\label{def:overshadowEdge}
Let $\theta, \gamma \in \Theta$. We say there is an overshadowing edge
$\theta \to \gamma$ if, $i \geq j, \text{ with } \theta \in \Theta_{\phi(\theta)}(i)  \text{ and } \gamma \in \Theta_{\phi(\theta)}(j)$.

\end{definition}

If $\theta \to \gamma$ in this sense, then whenever the algorithm samples 
action $\phi(\theta)$ infinitely often while $\theta$ has nonnegligible mass, 
the ratio $\pi_t(\gamma)/\pi_t(\theta)$ typically grows unbounded, 
so $\theta$ is eventually ``ruled out'' by $\gamma$. 
By contrast, if the sampling or posterior updates \emph{switch away} from 
$\phi(\theta)$ too soon, $\theta$ might never be fully overshadowed in practice.

We now view $\Theta$ as the vertex set of a directed graph $G=(\Theta,E)$ 
whose edges are precisely the overshadowing edges $\theta \to \gamma$.  
A subset of parameters $S\subseteq\Theta$ induces a subgraph $G[S]$.

\begin{definition}[Closed, Strongly Connected Subgraph]
A set $S\subseteq\Theta$ is called \emph{closed} under overshadowing if no vertex 
in $S$ has an outgoing edge in $G$ to a vertex outside $S$. That is,
$$
  \forall\,\theta\in S,\;\;\text{if}\;\theta\to\gamma,\;\text{then}\;
   \gamma\in S.
$$
We say $S$ is \emph{strongly connected} under overshadowing 
if for every pair $\theta,\gamma \in S$, there is a directed path 
$\theta \to \cdots \to \gamma$ within $G[S]$. 
\end{definition}

\begin{definition}[Jointly Closed Under Overshadowing]
    \label{def:JointlyClosed}
    A subset $S \subseteq \Theta$ is called \emph{jointly closed} under overshadowing 
    if there exists a stable posterior distribution 
    $p(\cdot)$ supported on $S$ such that, for every 
    $\gamma \notin S$, $\gamma$ cannot achieve a strictly positive 
    log‐likelihood ratio against $p(\cdot)$ on any action 
    $a \in \mathcal{A}^o(S)$.
    
    Equivalently, $S$ is jointly closed if there exists $p(\cdot)$ on $S$ 
    satisfying:
    $$
        \sum_{\theta \in S} p(\theta)
            \;\mathbb{E}\Bigl[\log \Lambda_s^{\gamma,\theta} \,\Bigm|\,
            A_s = a\Bigr]
        \;\le\; 0
        \quad
        \text{for all } 
        \gamma\notin S 
        \text{ and } 
        a \in \mathcal{A}^o(S).
    $$
\end{definition}

\subsection{Necessary Conditions for Collective Concentration}

As discussed, a bandit or Bayesian‐learning process can concentrate its posterior on 
a subset of parameters $S \subseteq \Theta$ along a given sample path 
(\emph{pathwise}), but there may be multiple such subsets that arise 
on different sample paths due to stochasticity in rewards. 
Hence, while we might establish conditions under which some run 
converges to $S$, the expected posterior across all runs 
could still be split among multiple such sets.

\begin{definition}[Pathwise Collective Concentration]
\label{def:PathwiseCollectiveConcentration}
We say the posterior collectively concentrates on $S\subseteq\Theta$ 
pathwise if there is a positive‐probability event $\Omega_S$ such that, 
on every sample path $\omega \in \Omega_S$, 
$$
  \lim_{t\to\infty} \sum_{\theta \in S}\pi_t(\theta)(\omega)
  \;=\; 1.
$$
That is, with positive probability, the posterior eventually allocates total mass 
1 to $S$ in a single run.
\end{definition}

The next proposition provides a necessary condition for this pathwise collective concentration, that is $S$ must form 
a jointly closed and strongly connected component in a suitable 
overshadowing graph. We will then show how the 
expected posterior might combine these different pathwise attractors.

\begin{proposition}[Necessary Conditions for Pathwise Collective Concentration]
    \label{prop:PathwiseCollectiveConcNecessary}
    Suppose there is a subset $S \subseteq \Theta$ and a positive‐probability event 
    $\Omega_S$ such that, for every sample path $\omega \in \Omega_S$, 
    $$
      \lim_{t\to\infty} \sum_{\theta\in S}\pi_t(\theta)(\omega) \;=\; 1.
    $$
    Then the following conditions must hold:
    \begin{enumerate}
    \item \textbf{Jointly Closed:} 
    For every parameter $\gamma \notin S$, given the distribution of actions 
    which are played infinitely often on $\Omega_S$, $\gamma$ is overshadowed
    by the collective mixture of parameters in $S$.

    \item \textbf{Internal Cohesion:}
    Within $S$, no parameter is permanently ruled out by the others. 
    Equivalently, for every pair $\theta, \gamma \in S$, there is a directed path 
    $\theta \to \cdots \to \gamma$ in the induced subgraph $G[S]$.  
    \end{enumerate}
    \end{proposition}
    
    \begin{proof}[Sketch of Proof]
    \textbf{(1) Jointly Closed.} 
    Since $S$ collectively captures total posterior mass $1$ on $\Omega_S$, 
    any parameter $\gamma\notin S$ that could ``beat'' all members of $S$ on the actions sampled infinitely often would eventually claim some positive fraction of posterior.\footnote{Note we do not require the outsider parameter $\gamma$ to overshadow any parameter in $S$. } Concretely, if repeated data on the event $\Omega_S$ were to show a sustained likelihood‐ratio growth favoring $\gamma$ relative to \emph{every} $\theta\in S$, then $\gamma$ would not remain at posterior mass near $0$. Hence $\sum_{\theta\in S}\pi_t(\theta)\to 1$ would be contradicted. Therefore, no external parameter $\gamma$ can systematically overtake 
    the collective mixture of parameters in $S$.\footnote{
    Even if $S$ is jointly closed, the posterior need not settle on a single 
    stable distribution within $S$; it may oscillate or cycle among the parameters 
    in $S$. One can thus generalize from a unique stable posterior distribution 
    to a broader stable region of distributions. Small changes in 
    overshadowing relationships or model parameters can then merge or split such 
    regions, producing a qualitative shift in the limiting behavior, akin to a 
    ``phase transition.''}

    \smallskip

    \textbf{(2) Internal Cohesion.}
    If a parameter $\theta\in S$ were definitively overshadowed 
    by another $\gamma\in S$ on infinitely many samples of $\phi(\theta)$, 
    $\theta$ would eventually lose its posterior mass for good. 
    One would prune out parameters from $S$ 
    until no nontrivial cycle remains. Then $S$ could not retain all the mass. 
    Hence every $\theta\in S$ must remain viable, implying $S$ is strongly connected 
    (i.e., no irreversible overshadowing within $S$).
    
    \smallskip
    
    Combining these points, $S$ must be robust against external parameters 
    and be internally ``cohesive'' while 
    its favored actions are indeed sampled infinitely often. 
    Thus the posterior can concentrate on $S$ pathwise with positive probability.
    \end{proof}

\begin{remark}[Relation to Expected Posterior]
There may be multiple disjoint subsets 
$S_1, \dots, S_k \subseteq\Theta$ each satisfying 
Definition~\ref{def:PathwiseCollectiveConcentration} on disjoint events 
with positive probability. 
Thus, in expectation, the limiting posterior could be a mixture over these sets. 
One might then write:
$$
   \lim_{t\to\infty} \mathbb{E}\Bigl[\sum_{\theta \in S_1}\!\pi_t(\theta)\Bigr] 
   \;=\; p_1, 
   \quad
   \lim_{t\to\infty} \mathbb{E}\Bigl[\sum_{\theta \in S_2}\!\pi_t(\theta)\Bigr] 
   \;=\; p_2,
   \quad
   \dots
   \quad
   \text{with }p_1 + \cdots + p_k = 1.
$$
In such a case, no single $S_i$ necessarily collects mass $1$ on average. 
Instead, each $S_i$ captures the entire posterior pathwise
on a subset of runs with probability $p_i$, 
giving a mixed limit for the expected posterior.
\end{remark}

\section{Conclusion and Future Work}

In this paper, we examined the behavior of dynamic decision-making algorithms under model misspecification in a finite-parameter setting. Our main contributions include: First, we introduced an extension of White's (1982) pseudo-true parameter concept for the dynamic case. Instead of a single pseudo-true point, we obtained a \emph{pseudo-truth set} \(\Theta^\dagger \subseteq \Theta\) whose elements provide equally good fits in terms of KL divergence. Second, we established that Thompson Sampling allocates posterior probability to parameters in \(\Theta^\dagger\) at an exponential rate. Moreover, the total regret grows linearly in time \(T\), implying a constant per-period regret in the limit. Third, within \(\Theta^\dagger\), multiple parameters may remain positive expeceted posterior forever, and the posterior can exhibit cyclic behavior among them. We derived necessary conditions under which pathwise posterior concentration on a subset of \(\Theta^\dagger\) can occur. Specifically, a subset \(S\subset \Theta^\dagger\) must be jointly closed under overshadowing and strongly connected to sustain nontrivial posterior mass in the long run. Finaly, our findings suggest that while exact parameter inference can fail under misspecification, but the simulation results shows Thompson Sampling can still yield robust action-selection performance. This suggests there may be structural conditions under which a Bayesian updating scheme remains near-optimal even when the assumed model is somewhat misspecified.

\newpage
\printbibliography

@InProceedings{pmlr-v119-lattimore20a,
  title = 	 {Learning with Good Feature Representations in Bandits and in {RL} with a Generative Model},
  author =       {Lattimore, Tor and Szepesvari, Csaba and Weisz, Gellert},
  booktitle = 	 {Proceedings of the 37th International Conference on Machine Learning},
  pages = 	 {5662--5670},
  year = 	 {2020},
  editor = 	 {III, Hal Daumé and Singh, Aarti},
  volume = 	 {119},
  series = 	 {Proceedings of Machine Learning Research},
  publisher =    {PMLR}
}

@inproceedings{NEURIPS2021_177db6ac,
 author = {Bogunovic, Ilija and Krause, Andreas},
 booktitle = {Advances in Neural Information Processing Systems},
 editor = {M. Ranzato and A. Beygelzimer and Y. Dauphin and P.S. Liang and J. Wortman Vaughan},
 pages = {3004--3015},
 publisher = {Curran Associates, Inc.},
 title = {Misspecified Gaussian Process Bandit Optimization},
 volume = {34},
 year = {2021}
}

@misc{li2023dynamicselectionalgorithmicdecisionmaking,
      title={Dynamic Selection in Algorithmic Decision-making}, 
      author={Jin Li and Ye Luo and Xiaowei Zhang},
      year={2023},
      eprint={2108.12547},
      archivePrefix={arXiv},
      primaryClass={econ.EM},
}

@inproceedings{NEURIPS2020_84c230a5,
 author = {Foster, Dylan J and Gentile, Claudio and Mohri, Mehryar and Zimmert, Julian},
 booktitle = {Advances in Neural Information Processing Systems},
 editor = {H. Larochelle and M. Ranzato and R. Hadsell and M.F. Balcan and H. Lin},
 pages = {11478--11489},
 publisher = {Curran Associates, Inc.},
 title = {Adapting to Misspecification in Contextual Bandits},
 volume = {33},
 year = {2020}
}

@article{bonhomme2022minimizing,
  title={Minimizing sensitivity to model misspecification},
  author={Bonhomme, St{\'e}phane and Weidner, Martin},
  journal={Quantitative Economics},
  volume={13},
  number={3},
  pages={907--954},
  year={2022},
  publisher={Wiley Online Library}
}

@article{white1982maximum,
  title={Maximum likelihood estimation of misspecified models},
  author={White, Halbert},
  journal={Econometrica: Journal of the Econometric Society},
  pages={1--25},
  year={1982},
  publisher={JSTOR}
}

@article{masten2021salvaging,
  title={Salvaging falsified instrumental variable models},
  author={Masten, Matthew A and Poirier, Alexandre},
  journal={Econometrica},
  volume={89},
  number={3},
  pages={1449--1469},
  year={2021},
  publisher={Wiley Online Library}
}

@article{mullerlocally,
  title={Locally robust efficient bayesian inference},
  author={M{\"u}ller, ULRICH and Norets, ANDRIY},
  year={2024}
}

@misc{ba2023robustmisspecifiedmodelsparadigm,
      title={Robust Misspecified Models and Paradigm Shifts}, 
      author={Cuimin Ba},
      year={2023},
      eprint={2106.12727},
      archivePrefix={arXiv},
      primaryClass={econ.TH},
}

@article{ESPONDA2021105260,
title = {Asymptotic behavior of Bayesian learners with misspecified models},
journal = {Journal of Economic Theory},
volume = {195},
pages = {105260},
year = {2021},
issn = {0022-0531},
author = {Ignacio Esponda and Demian Pouzo and Yuichi Yamamoto}
}

@article{https://doi.org/10.3982/ECTA12609,
author = {Esponda, Ignacio and Pouzo, Demian},
title = {Berk–Nash Equilibrium: A Framework for Modeling Agents With Misspecified Models},
journal = {Econometrica},
volume = {84},
number = {3},
pages = {1093-1130},
year = {2016}
}

@article{fudenberg2023misspecifications,
  title={Which misspecifications persist?},
  author={Fudenberg, Drew and Lanzani, Giacomo},
  journal={Theoretical Economics},
  volume={18},
  number={3},
  pages={1271--1315},
  year={2023},
  publisher={Wiley Online Library}
}

@article{andrews2023structural,
  title={Structural estimation under misspecification: theory and implications for practice},
  author={Andrews, Isaiah and Barahona, Nano and Gentzkow, Matthew and Rambachan, Ashesh and Shapiro, Jesse M},
  year={2023}
}

@techreport{armstrong2024adapting,
  title={Adapting to Misspecification},
  author={Armstrong, Timothy and Kline, Patrick M and Sun, Liyang},
  year={2024},
  institution={National Bureau of Economic Research}
}

@article{kim2017thompson,
  title={Thompson sampling for stochastic control: The finite parameter case},
  author={Kim, Michael Jong},
  journal={IEEE Transactions on Automatic Control},
  volume={62},
  number={12},
  pages={6415--6422},
  year={2017},
  publisher={IEEE}
}

@article{fan2021diffusion,
  title={Diffusion Approximations for Thompson Sampling},
  author={Fan, Lin and Glynn, Peter W},
  journal={arXiv preprint arXiv:2105.09232},
  year={2021}
}

@article{adusumilli2021risk,
  title={Risk and optimal policies in bandit experiments},
  author={Adusumilli, Karun},
  journal={arXiv preprint arXiv:2112.06363},
  year={2021}
}

@article{salant2020statistical,
  title={Statistical inference in games},
  author={Salant, Yuval and Cherry, Josh},
  journal={Econometrica},
  volume={88},
  number={4},
  pages={1725--1752},
  year={2020},
  publisher={Wiley Online Library}
}

@techreport{murooka2023higher,
  title={Higher-order Misspecification and Equilibrium Stability},
  author={Murooka, Takeshi and Yamamoto, Yuichi},
  year={2023},
  institution={Osaka School of International Public Policy, Osaka University}
}

@article{liang2019games,
  title={Games of incomplete information played by statisticians},
  author={Liang, Annie},
  journal={arXiv preprint arXiv:1910.07018},
  year={2019}
}

@article{banjevic2019thompson,
  title={Thompson sampling for stochastic control: The continuous parameter case},
  author={Banjevi{\'c}, Dragan and Kim, Michael Jong},
  journal={IEEE Transactions on Automatic Control},
  volume={64},
  number={10},
  pages={4137--4152},
  year={2019},
  publisher={IEEE}
}

@book{lattimore2020bandit,
  title={Bandit algorithms},
  author={Lattimore, Tor and Szepesv{\'a}ri, Csaba},
  year={2020},
  publisher={Cambridge University Press}
}

@article{chung2006concentration,
  title={Concentration inequalities and martingale inequalities: a survey},
  author={Chung, Fan and Lu, Linyuan},
  journal={Internet mathematics},
  volume={3},
  number={1},
  pages={79--127},
  year={2006},
  publisher={Taylor \& Francis}
}

\newpage
\section{Appendix}

\subsection{Proof of Proposition 1}
\begin{proof}

    This proof is an one-step generalization of the approach by \textcite{kim2017thompson}. We can firstly re-write the posterior distribution of the pseudo-truth parameter set $\Theta^{\dagger} \subset \Theta$ as:

    \begin{equation}
        \begin{aligned}
            \pi_t(\Theta^{\dagger}) & =\frac{\sum_{\theta \in \Theta^\dagger} \mathcal{L}_\theta\left(H_t\right) \pi_0(\theta)}{\sum_{\gamma \in \mathcal{P}} \mathcal{L}_\gamma\left(H_t\right) \pi_0(\gamma)} \\
            & = \frac{ \mathcal{L}_{\Theta^\dagger}\left(H_t\right) \pi_0(\Theta^\dagger)}{\sum_{\gamma \in \mathcal{P}} \mathcal{L}_\gamma\left(H_t\right) \pi_0(\gamma)} \\
            & = \frac{1}{1+\sum_{\gamma \notin \Theta^\dagger} c_\gamma\left(\frac{\mathcal{L}_\gamma\left(H_t\right)}{\mathcal{L}_{\Theta^\dagger}\left(H_t\right)}\right)} \\
            & = \frac{1}{1+\sum_{\gamma \notin \Theta^\dagger} c_\gamma \exp \left(-\sum_{s=0}^t \log \Lambda_s^\gamma\right)}
        \end{aligned}
    \end{equation}

    where the constant $c_\gamma = \frac{\pi_0(\gamma)}{\pi_0(\Theta^\dagger)}$ and $\Lambda_s^\gamma = \frac{\mathcal{L}_{\Theta^\dagger}\left(H_s\right)}{\mathcal{L}_\gamma\left(H_s\right)}$. Here $\pi_0(\Theta^\dagger)$ is the prior probability of the pseudo-truth set $\Theta^\dagger$, which is defined as $\pi_0(\Theta^\dagger) = \sum_{\theta \in \Theta^\dagger} \pi_0(\theta)$ in the finite parameter space setting. And the $\mathcal{L}_{\Theta^\dagger}\left(H_t\right)$ is the likelihood of the pseudo-truth set $\Theta^\dagger$ given the history $H_t$, which is defined as $\mathcal{L}_{\Theta^\dagger}\left(H_t\right) = \sum_{\theta \in \Theta^\dagger} \mathcal{L}_\theta\left(H_t\right)\tilde{\pi}_0(\theta)$ with normalized prior $\tilde{\pi}_0(\theta) = \frac{\pi_0(\theta)}{\pi_0(\Theta^\dagger)}$.

    Then we can define the following stochastic process:

    \begin{equation}
        Z_t = \sum_{s=0}^t \log \Lambda_s^\gamma
    \end{equation}

    By Doob Decomposition, we can write $Z_t$ as:

    \begin{equation}
        \begin{aligned}
            Z^\gamma_t & = M^\gamma_t + I^\gamma_t \\
            & = \left(\sum_{s=0}^t \log \Lambda_s^\gamma - \sum_{s=0}^t \mathbb{E}_{\theta^*}\left[\log \Lambda_s^\gamma \mid \mathcal{H}_{s-1}\right] \right) + \sum_{s=0}^t \mathbb{E}_{\theta^*}\left[\log \Lambda_s^\gamma \mid \mathcal{H}_{s-1}\right] 
        \end{aligned}
    \end{equation}

    This is a sub-martingale respect to the filtration $\mathcal{H}_t = \sigma(H_t)$. Note that the conditional expectation is taken with respect to the true parameter $\theta^*$ instead of $\theta$.

    \begin{lemma}\label{lemma:finiteIncrements}
        The martingale $M_t^\gamma$ has finite increments. That is there exist a constant $d > 0$ such that:
        \begin{equation}
            | \log \Lambda_s^\gamma - \mathbb{E}_{\theta^*}\left[\log \Lambda_s^\gamma \mid \mathcal{H}_{s-1}\right] |\leq d
        \end{equation}
    \end{lemma}



    \begin{proof}
        By the assumption that our action space $\mathcal{A}$ is compact, so the function $f_{\gamma}(r \mid A_t)$ can get maximun and minimun value. This ensures $\log \Lambda_s^\gamma$ is always positive and finite.
    \end{proof}

    \begin{lemma}\label{lemma:boundedIncrements}
    The second term $I_t^\gamma$ is a predictable process. That is there exist a constant $c > 0$ such that:

    \begin{equation}
        \mathbb{E}_{\theta^*}\left[\log \Lambda_s^\gamma \mid \mathcal{H}_{s-1}\right] \leq c
    \end{equation}
    \end{lemma}

    \begin{proof}
        
    \begin{equation}
        \begin{aligned}
            & \mathbb{E}_{\theta^*}^\tau\left[\log \Lambda_s^\gamma \mid \mathcal{H}_{s-1}\right] \\
            = & \mathbb{E}_{\theta^*}^\tau\left[\mathbb{E}_{\theta^*}^\tau\left[\log \Lambda_s^\gamma \mid \mathcal{H}_{s-1}, A_{s-1}\right] \mid \mathcal{H}_{s-1}\right] \\
            = & \mathbb{E}_{\theta^*}^\tau\left[\mathcal{K}\left(v_{\Theta^\dagger}\left(\cdot \mid X_{s-1}, A_{s-1}\right) \mid v_\gamma\left(\cdot \mid X_{s-1}, A_{s-1}\right)\right) \mid \mathcal{H}_{s-1}\right] \\
            \geq & \mathbb{E}_{\theta^*}^\tau\left[\min _{x \in \mathcal{X}, a \in \mathcal{A}} \mathcal{K}\left(v_{\Theta^\dagger}^{x, a} \mid v_\gamma^{x, a}\right) \mid \mathcal{H}_{s-1}\right] \\
            \geq & \mathbb{E}_{\theta^*}^\tau\left[\max _{\theta \in \Theta^\dagger} \inf _{x \in \mathcal{X}, a \in \mathcal{A}} \mathcal{K}\left(v_{\theta}^{x, a} \mid v_\gamma^{x, a}\right) \mid \mathcal{H}_{s-1}\right] \\
            \geq & \min_{\theta \neq \gamma \in \Theta} \max _{\theta \in \Theta^\dagger} \inf_{x \in \mathcal{X}, a \in \mathcal{A}} \epsilon(x, a, \theta, \gamma) \\
            = & \epsilon>0
        \end{aligned}
        \end{equation}
        
        The last inequality comes from the fact of assumption 2. 
    \end{proof}

    Hence we can write the expectation of the posterior of the pseudo-truth set as:

    \begin{equation}
        \begin{aligned}
        \mathbb{E}_{\theta^*}^\tau \left[ \pi_t(\Theta^{\dagger}) \right] & = 
        \mathbb{E}_{\theta^*}^\tau \left[ \frac{1}{1 - \sum_{\gamma \notin \Theta^\dagger} c_\gamma (M_t^\gamma + I_t^\gamma)} \right] \\
            & \geq \mathbb{E}_{\theta^*}^\tau \left[ \frac{1}{1 - \sum_{\gamma \notin \Theta^\dagger} c_\gamma (M_t^\gamma + \epsilon t)} \right] 
        \end{aligned}
    \end{equation}

    Then following the same methods by \textcite{kim2017thompson}, for any $\delta>0$ and $\gamma \in \mathcal{P}$ we can define the event:

    $$
    B_t^\gamma(\delta)=\left\{\left|M_t^\gamma\right| \leq \delta t\right\}
    $$

    then, for any choice of $0<\delta<\epsilon$

    $$
    \begin{aligned}
    & \mathbb{E}_{\theta^*}^\tau\left[\pi_t(\Theta^{\dagger})\right] \\
    \geq & \mathbb{E}_{\theta^*}^\tau\left[\frac{1}{1+\sum_{\gamma \notin \Theta^\dagger} c_\gamma \exp \left(-M_t^\gamma-\epsilon t\right)} \mathbbm{1} \cap_{\gamma \notin \Theta^\dagger} B_t^\gamma(\delta)\right] \\
    & +\mathbb{E}_{\theta^*}^\tau\left[\frac{1}{1+\sum_{\gamma \notin \Theta^\dagger} c_\gamma \exp \left(-M_t^\gamma-\epsilon t\right)} \mathbbm{1}\left(\cap_{\gamma \notin \Theta^\dagger} B_t^\gamma(\delta)\right)^c\right] \\
    \geq & \frac{\mathbb{P}_{\theta^*}^\tau\left(\cap_{\gamma \notin \Theta^\dagger} B_t^\gamma(\delta)\right)}{1+\frac{1-\pi_0(\Theta^{\dagger})}{\pi_0(\Theta^{\dagger})} \exp (-(\epsilon-\delta) t)} \\
    = & \frac{1-\mathbb{P}_{\theta^*}^\tau\left(\cup_{\gamma \notin \Theta^\dagger} B_t^\gamma(\delta)^c\right)}{1+\frac{1-\pi_0(\Theta^{\dagger})}{\pi_0(\Theta^{\dagger})} \exp (-(\epsilon-\delta) t)} .
    \end{aligned}
    $$

    The second inequality is because inside the second expectation, the term is always positive and greater than zero. Then by the union bound we have:

    $$
    \begin{aligned}
    \mathbb{E}_{\theta^*}^\tau\left[\pi_t(\Theta^\dagger)\right] & \geq \frac{1-\sum_{\gamma \notin \Theta^\dagger} \mathbb{P}_{\theta^*}^\tau\left(B_t^\gamma(\delta)^c\right)}{1+\frac{1-\pi_0(\Theta^{\dagger})}{\pi_0(\Theta^{\dagger})} \exp (-(\epsilon-\delta) t)} \\
    & \geq \frac{1-\sum_{\gamma \notin \Theta^\dagger} \mathbb{P}_{\theta^*}^\tau\left(\left|M_t^\gamma\right| \geq \delta t\right)}{1+\frac{1-\pi_0(\Theta^{\dagger})}{\pi_0(\Theta^{\dagger})} \exp (-(\epsilon-\delta) t)}
    \end{aligned}
    $$

By the Lemma \ref{lemma:boundedIncrements}, we can have:

\begin{equation}
    \mathbb{P}_{\theta^*}^\tau\left(\left|M_t^\gamma\right| \geq \delta t\right) \leq 2 \exp \left(-\frac{\delta^2 t^2}{2 d^2}\right)
\end{equation}

The inequality is given by the Azuma-Hoeffding inequality (\cite{chung2006concentration}). Then we can have:

\begin{equation}
    \mathbb{E}_{\theta^*}^\tau\left[\pi_t(\Theta^\dagger)\right] \geq \frac{1-2(|\mathcal{P}|-|\Theta^\dagger|) \exp \left(-\frac{\delta^2 t}{2 d^2}\right)}{1+\frac{1-\pi_0(\Theta^\dagger)}{\pi_0(\Theta^\dagger)} \exp (-(\epsilon-\delta) t^2)}
\end{equation}

Then we can get:

\begin{equation}\label{eq:lower_bound}
    \begin{aligned}
    & \mathbb{E}_{\theta^*}^\tau\left[1-\pi_t(\Theta^\dagger)\right] \\
    = & \frac{\frac{1-\pi_0(\Theta^\dagger)}{\pi_0(\Theta^\dagger)} \exp (-(\epsilon-\delta) t)+2(|\mathcal{P}|-|\Theta^\dagger|) \exp \left(-\frac{\delta^2 t^2}{2 d^2}\right)}{1+\frac{1-\pi_0(\Theta^\dagger)}{\pi_0(\Theta^\dagger)} \exp (-(\epsilon-\delta) t)} \\
    \leq & \frac{1-\pi_0(\Theta^\dagger)}{\pi_0(\Theta^\dagger)} \exp (-(\epsilon-\delta) t)+2(|\mathcal{P}|-|\Theta^\dagger|) \exp \left(-\frac{\delta^2 t^2}{2 d^2}\right) \\
    = & \frac{1-\pi_0(\Theta^\dagger)}{\pi_0(\Theta^\dagger)} \exp \left(-\frac{\epsilon t}{2}\right)+2(|\mathcal{P}|-|\Theta^\dagger|) \exp \left(-\frac{\epsilon^2 t^2}{8 d^2}\right) \\   \leq & \frac{1-\pi_0(\Theta^\dagger)}{\pi_0(\Theta^\dagger)} \exp \left(-\frac{\epsilon t}{2}\right)+2(|\mathcal{P}|-|\Theta^\dagger|) \exp \left(-\frac{\epsilon^2 t}{8 d^2}\right) \\
    \leq & a_\theta \exp \left(-b_\theta t\right)
    \end{aligned}
\end{equation}

One possible choice can be $a_{\Theta^\dagger} =
2 \max \left\{ \frac{1 - \pi_0(\Theta^{\dagger})}{\pi_0(\Theta^{\dagger})}, 2(|\mathcal{P}| - |\Theta^\dagger|) \right\} $
and $b_{\Theta^\dagger}  = \min \left\{ \frac{\epsilon}{2}, \frac{\epsilon^2}{8d^2} \right\}$. \footnote{From the proof we know that for sufficiently large $t$, first term in the 
\eqref{eq:lower_bound} will dominate the second term. Hence the rate will be approximately agree with the constants as $a_{\Theta^\dagger} = \frac{1 - \pi_0(\Theta^{\dagger})}{\pi_0(\Theta^{\dagger})}$ and  $b_{\Theta^\dagger}  = \frac{\epsilon}{2}$.}

\end{proof}

\end{document}